\documentclass[11pt,a4paper]{article}
\usepackage{fullpage}
\usepackage[utf8]{inputenc}
\usepackage{amsmath,amssymb,mathtools}
\usepackage{amsthm}
\usepackage{xcolor}
\usepackage{thmtools}
\usepackage{subcaption,thm-restate}
\usepackage[sort]{cite}
\usepackage{titling}
\usepackage{hyperref}
\hypersetup{
    colorlinks=true,       
    linkcolor=blue,        
    citecolor=red,         
    filecolor=magenta,     
    urlcolor=cyan,         
    linktocpage=true
}
\usepackage{graphicx}
\usepackage{cleveref}
\usepackage{algorithm}
\usepackage{algorithmicx}
\usepackage[noend]{algpseudocode}
\usepackage{tikz}
\usetikzlibrary {calc,decorations.pathmorphing, patterns,positioning} 
\usepackage[title]{appendix}

\definecolor{greenish}{RGB}{27,158,119}
\definecolor{MyOrange}{RGB}{230,97,1}
\definecolor{MyPurple}{RGB}{94,60,153}
\definecolor{myLightPurple}{RGB}{178,171,210}

\def\NP{\mathsf{NP}}

\makeatletter
\newcommand{\algmargin}{\the\ALG@thistlm}
\makeatother
\algnewcommand{\parState}[1]{\State%
    \parbox[t]{\dimexpr\linewidth-\algmargin}{\strut\hangindent=\algorithmicindent \hangafter=1 #1\strut}}

\theoremstyle{plain}
\newtheorem{thm}{Theorem}
\newtheorem{lem}[thm]{Lemma}

\theoremstyle{definition}
\newtheorem{defn}[thm]{Definition}
\newtheorem{rem}[thm]{Remark}

\usepackage{stackengine} %
\newcommand\cuparrow{\ensurestackMath{\stackinset{c}{0ex}{c}{0.2ex}{\scriptscriptstyle \downarrow}{\cup}}}
\newcommand\bigcuparrow{\ensurestackMath{\stackinset{c}{0ex}{c}{0.2ex}{\scriptscriptstyle \downarrow}{\bigcup}}}

\def\SOP{\textsc{Shortest Odd Path}}
\def\DISP{\textsc{Shortest Two Disjoint Paths}}
\def\SPaPrOP{\textsc{Shortest Parity-Constrained Odd Path}}
\def\sq{\textup{sqn}}

\def\even{\mathrm{even}}
\def\odd{\mathrm{odd}}
\def\T{\mathcal{T}}

\title{Shortest odd paths in undirected graphs\\ with conservative weight functions}

\thanksmarkseries{alph}
\author{
\hfill
\and
Alp\'ar J\"uttner\thanks{ELKH-ELTE Egerváry Research Group, Eötvös Loránd Research Network (ELKH), Budapest, Hungary.} \thanks{Department of Operations Research, ELTE Eötvös Loránd University, Budapest, Hungary. Email: \texttt{alpar@cs.elte.hu,csaba.kiraly@ttk.elte.hu,lyd21@student.elte.hu,gyula.pap@ttk.elte.hu}.}
\and
Csaba Kir\'aly\footnotemark[1] \footnotemark[2] 
\and
Lydia Mirabel Mendoza-Cadena\footnotemark[2] \thanks{MTA-ELTE Momentum Matroid Optimization Research Group, Budapest, Hungary.} 
\and
Gyula Pap\footnotemark[1] \footnotemark[2]
\and
Ildik\'o Schlotter\thanks{Centre for Economic and Regional Studies, and also Budapest University of Technology and Economics, Budapest, Hungary. Email: \texttt{schlotter.ildiko@krtk.hu}.}
\and
Yutaro Yamaguchi\thanks{Department of Information and Physical Sciences, Graduate School of Information Science and Technology, Osaka University, Osaka, Japan. Email: \texttt{yutaro.yamaguchi@ist.osaka-u.ac.jp}.}
}
\date{ }

\begin{document}

\maketitle

\begin{abstract}
    We consider the \SOP{} problem, where given an undirected graph $G$, a weight function on its edges, and two vertices $s$ and $t$ in~$G$, the aim is to find an $(s,t)$-path with odd length and, among all such paths, of minimum weight. For the case when the weight function is conservative, i.e., when every cycle has non-negative total weight, the complexity of the \SOP{} problem had been open for 20 years, and was recently shown to be $\NP$-hard. 
    
    We give a polynomial-time algorithm for the special case when the weight function is conservative and the set~$E^-$ of negative-weight edges forms a single tree. Our algorithm exploits the strong connection between \SOP{}
    and the problem of finding two internally vertex-disjoint paths between two terminals in an undirected edge-weighted graph. 
    It also relies on solving an intermediary problem variant called \SPaPrOP{}  
    where for certain edges we have parity constraints on their position along the path.

    Also, 
    we exhibit two FPT algorithms for solving \SOP{} in graphs with conservative weight functions. 
    The first FPT algorithm is parameterized by $|E^-|$, the number of negative edges, or more generally, by the maximum size of a matching in the subgraph of~$G$ spanned by~$E^-$. 
    Our second FPT algorithm is parameterized by the treewidth of~$G$.

    \medskip
    

\end{abstract}

\section{Introduction} \label{sec:intro}
Lov\'{a}sz~\cite{schrijver-book} asked for the complexity of the following question nearly 20 years ago: 
Given two vertices~$s$ and~$t$ in an undirected graph $G$ and a weight function~$w$ on the edges of~$G$ such that each cycle in~$G$ has non-negative weight, find a minimum-weight odd path between~$s$ and~$t$. Throughout the paper, a path is \emph{odd} (or even) if it contains an odd (even, respectively) number of edges.

Recently, the problem was shown to be $\NP$-hard
by Schlotter and Sebő~\cite{schlotter2022shortest}, even when the weight function is conservative and takes values only from~$\{-1,1\}$. 
Note that replacing ``odd'' with ``even'' does not change the computational complexity of the problem: 
finding a minimum-weight odd (or even) $(s,t)$-path (i.e., 
a path between $s$ and~$t$) in~$G$ can easily be reduced to finding a minimum-weight even (or 
odd, respectively) path between $s$ and a newly added vertex~$t'$ in the graph obtained from~$G$ by adding an edge $tt'$ with weight~$0$. 
Note also that finding a shortest odd (even) path between two vertices $s$ and~$t$ generalizes the shortest path problem: 
adding new vertices $t_1$ and~$t_2$ along with 0-weight edges $tt_2$, $tt_1$, and $t_1t_2$, a shortest odd (or even, resp.) $(s,t_2)$-path in the resulting graph yields a shortest $(s,t)$-path in the original graph. 
For unweighted graphs with possibly negative but \emph{conservative} edge-weights (meaning that there are no cycles with negative total weight), the shortest path problem is solvable in polynomial time as shown by Seb\H{o} in \cite{Sebo86}; however, the usual shortest path algorithms for directed graphs are not applicable, since bidirecting an edge with negative weight creates a negative cycle.

If the weights are non-negative, a minimum-weight odd path between fixed vertices can be found via computing a maximum-weight matching problem in an auxiliary graph. Thomassen attributes this algorithm to Edmonds~\cite{Thomassen85}, while Gr\"otschel and Pulleybank refer to it as ``Waterloo folklore''~\cite{GP81}. 
A Dijsktra-like algorithm was given for the problem by Derigs~\cite{derigs1985shortest}. 
For the unweighted case, LaPaugh and Papadimitriou \cite{lapaugh1984even} gave a linear-time algorithm that computes an odd (or 
even) path between two vertices, whenever such a path exists;
Arkin et al.~\cite{APY91} achieved linear-time solvability even for finding a shortest odd (or even) path between two vertices.

The generalized problem of finding a shortest non-zero path in group-labeled graphs with non-negative edge weights has also been studied. 
A recent work of Iwata and Yamaguchi \cite{yamaguchi2022groupLabeled} shows a deterministic and strongly polynomial-time algorithm for this problem.

For directed graphs, the problem becomes much harder: already deciding whether a path of prescribed parity exists between two given vertices is $\NP$-hard, as proved independently by Thomassen~\cite{Thomassen85} and by Lapaugh and Papadimitriou~\cite{lapaugh1984even}. 

The closely related problem of finding parity-constrained cycles is also of interest, with particular emphasis on odd holes, where a hole is an induced cycle of length at least four. Chudnovsky et al.~\cite{chudnovsky2020hole} recently showed that the problem of detecting if a graph contains an odd hole can be solved in polynomial time. Since then, their work inspired several algorithms in this area.

\paragraph{Our results.}
We consider the \SOP{} problem defined as follows:

\begin{center}
\fbox{ 
\parbox{13.6cm}{
\begin{tabular}{l}\SOP{}:  \end{tabular} \\
\begin{tabular}{p{1cm}p{11.5cm}}
\textbf{Input}: & An undirected graph $G=(V,E)$, a weight function~$w\colon E \to \mathbb{R}$, and two vertices $s$ and $t$ in~$G$. \\
\textbf{Goal}: & Find a minimum-weight odd $(s,t)$-path.
\end{tabular}
}}
\end{center}

We solve the \SOP{} problem in polynomial time when the weight function is conservative and the set~$E^-$ of negative edges forms a tree in~$G$. 
We also present two FPT algorithms for solving \SOP{} on graphs with conservative weights. 
In the first FPT algorithm, 
our parameter is~$|E^-|$, the number of negative edges; in fact, the approach we use can be extended to yield an FPT algorithm where the parameter is the size of a maximum matching in~$G[E^-]$, the subgraph of~$G$ spanned by~$E^-$. 
In the second FPT algorithm, the parameter is the treewidth of~$G$.

Our approach relies on solving a variant of the problem that we call \SPaPrOP . 
In order to define this problem, we give the following definition.
Let $F_\even \subseteq E$ and $F_\odd \subseteq E$ be two disjoint edge sets in~$G$.
An odd path $Q$ is \emph{$(F_\even,F_\odd)$-constrained}, if the sequence number of each edge in $Q \cap F_\even$ and in $Q \cap F_\odd$
is even and odd, respectively.
Here, the \emph{sequence number} of an edge~$e$ in~$Q$ is $i$, if $e$ is the $i$-th edge on~$Q$; note that since~$Q$ is odd, $e$ has an odd (even) sequence number independently whether we count the edges of~$Q$ starting from~$s$ or from~$t$.

Using this terminology, we consider the following problem.

\begin{center}
\fbox{ 
\parbox{13.6cm}{
\begin{tabular}{l}\SPaPrOP{}:  \end{tabular} \\
\begin{tabular}{p{1cm}p{11.5cm}}
\textbf{Input}: & An undirected graph $G=(V,E)$, a weight function~$w\colon E \to \mathbb{R}$, two vertices $s$ and $t$ in~$G$, and disjoint edge sets $F_\even\subseteq E$ and $F_\odd\subseteq E$. \\
\textbf{Goal}: & Find a minimum-weight $(F_\even,F_\odd)$-constrained odd $(s,t)$-path.
\end{tabular}
}}
\end{center}

We show that this problem can be solved in polynomial time if all negative-weight edges are contained in $F_\even\cup F_\odd$; 
this result will be useful for us when developing our algorithms for \SOP{}.
\medskip

\noindent
{\bf Organization.}
Basic definitions and notation are introduced in Section~\ref{sec:preliminaries}. Section~\ref{sec:tree-init} presents some initial observations and key definitions that will be useful in our algorithms. In Section~\ref{sec:parity}, we present a polynomial-time algorithm for the \SPaPrOP{}   problem. In Section \ref{sec:tree}, we show how to solve the \SOP{} problem when the weight function is conservative and the negative edges form a tree. Finally, in Section~\ref{sec:FPT}  we develop FPT algorithms for the \SOP{} problem.

\section{Preliminaries \label{sec:preliminaries}} 
\paragraph{Basic notation.} We denote the set of real numbers by $\mathbb{R}$, the rational numbers by $\mathbb{Q}$, and the non-negative real numbers by $\mathbb{R}_{\geq 0}$. 
Given subsets $X,Y\subseteq E$, the \emph{symmetric difference} of $X$ and $Y$ is denoted by $X\triangle Y\coloneqq (X\setminus Y)\cup(Y\setminus X)$. 

\paragraph{Graphs.} 
Let a graph~$G$ be a pair $(V,E)$ where $V$ and $E$ are the sets of vertices and edges, respectively.
For two vertices $u$ and $v$ in~$V$, an edge connecting $u$ and $v$ is denoted by $uv$ or $vu$. 
For $X\subseteq V$ and $F\subseteq E$, the graph obtained by deleting $X$ or $F$ is denoted by $G-X$ or $G-F$, respectively. $G[F]$ denotes the graph \emph{spanned} by $F$, that is, 
the subgraph of~$G$ containing only the edges of~$F$ and all the vertices incident to them.

A \emph{walk}~$W$ in~$G$ is a series $e_1, e_2, \dots, e_\ell$ of edges in~$G$ for which there exist vertices $v_0, v_1, \dots, v_\ell$ in~$G$ such that $e_i= v_{i-1} v_i$ for each $i \in \{1,2,\dots,\ell\}$; note that both vertices and edges may appear repeatedly on a walk, and in this sense we think of the sets of vertices and edges of a walk as multisets in which an element might appear more than once. 
In contrast,  we denote by $V(W)$ and $E(W)$ the sets of vertices and edges \emph{contained by} $W$, respectively, where each item is taken with multiplicity one.
We define the \emph{length} of~$W$ as $|W|=\ell$, and we say that $W$ is \emph{odd} (or \emph{even}) if $\ell$ is odd (or even, respectively). 
The \emph{endpoints} of $W$ are $v_0$ and $v_\ell$, or in other words, it is a  \emph{$(v_0,v_\ell)$-walk}, while all vertices on~$W$ that are not endpoints are \emph{inner vertices}.
If $v_0=v_\ell$, then we say that $W$ is a \emph{closed walk}.

A \emph{path} is a walk on which no vertex appears more than once.
We usually treat a path as a set $\{e_1,e_2,\dots, e_\ell\}$ of edges for which there exist distinct vertices $v_0, v_1, \dots, v_\ell$ in~$G$ such that $e_i= v_{i-1} v_i$ for each $i \in \{1,2,\dots,\ell\}$.
For any $i$ and $j$ with $0 \leq i \leq j \leq \ell$, we write $P[v_i,v_j]$ for the \emph{subpath} of $P$ between~$v_i$ and~$v_j$, consisting of edges $e_{i+1}, \dots, e_j$. 
Note that since we associate no direction with~$P$, we have $P[v_i,v_j]=P[v_j,v_i]$.
For indices $i_1, i_2, \dots, i_r$, we say that the vertices $v_{i_1}, v_{i_2}, \dots, v_{i_r}$ \emph{follow each other in this order} if $i_1 \leq i_2 \leq \dots \leq i_r$ or $i_1 \geq i_2 \geq \dots \geq i_r$.

The \emph{sequence number} of edge $e_i=v_{i-1} v_i \in P$  with respect to~$v_0$, denoted by $\sq(e_i,P,v_0)$, is $i$;
observe that the sequence number of $e_i$ w.r.t.~$v_\ell$ is~$\ell-i+1$.
Note that if~$P$ is an odd path, then the sequence number of an edge~$e$ on~$P$ w.r.t.~$v_0$ has the same parity as the sequence number of~$e$ w.r.t.~$v_\ell$. 
Therefore, in such a case we may talk about the parity of the sequence number of some edge~$e \in P$ without specifying from which endpoint of~$P$ we count the sequence number, and we may also write~$\sq(e,P)$ to denote the sequence number of~$e$ with respect to one (arbitrarily fixed) endpoint of~$P$.

Two paths~$P_1$ and~$P_2$ are \emph{openly disjoint} if they share neither edges nor vertices except for endpoints of both~$P_1$ and~$P_2$.
A \emph{cycle} is a set of edges that can be partitioned into two openly disjoint paths between the same pair of endpoints.

As mentioned before, if we have a non-negative weight function, then we can find a solution for the \SOP{} problem efficiently. Since finding a minimum-weight even path can be reduced to finding a minimum-weight odd path (as explained in Section~\ref{sec:intro}), we have the following.

\begin{thm}[\hspace{-1sp}\cite{Thomassen85, GP81}]\label{thm:SOP_nonnegative_weight}
    Given an undirected graph $G = (V,E)$ with vertices~$s$ and~$t$ and a non-negative weight function $w\colon E \to \mathbb{R}_{\geq 0}$, we can find a minimum-weight odd $(s,t)$-path and a minimum-weight even $(s,t)$-path in strongly polynomial time, whenever such paths exist.
\end{thm}

A set $T \subseteq E$ of edges in~$G$ is \emph{connected} if for every pair of edges $e$ and $e'$ in~$T$, there is a path contained in~$T$ containing both $e$ and $e'$.
If $T$ is connected and \emph{acyclic}, i.e., contains no cycle, then $T$ is a \emph{tree} in~$G$. 
Given two vertices $a$ and $b$ in a tree~$T$, we denote by $T[a,b]$ the unique path contained in~$T$ whose endpoints are $a$ and $b$.
For an edge $uv \in T$ and a path~$P$ in~$T$ such that $uv \notin P$, either $T[u,p]$ contains~$v$ for every vertex $p$ on~$P$, or $T[v,p]$ contains~$u$ for every vertex $p$ on~$P$.
In the former case, we say that $v$ \emph{is closer to~$P$} in~$T$ than~$u$.

We say that a weight function~$w\colon E \to \mathbb{R}$ 
is \emph{conservative}, if no cycle in~$G$ has negative weight. 
Given a weight function~$w\colon E\rightarrow \mathbb{R}$ on the edges of~$G$, the weight of an edge set~$F \subseteq E$ is $w(F)=\sum_{e \in F}w(f)$; similarly, the weight of a walk~$W=(e_1,\dots, e_\ell)$ is $w(W)=\sum_{i=1}^\ell w(e_i)$.

\paragraph{Derandomization and universal sets.} Given a randomized algorithm, a standard tool for  derandomization is the use of universal sets (see e.g., Cygan et al.~\cite{Cygan2015Parameterized}). An \emph{$(n,k)$-universal set} is a family $\mathcal{U}$ of subsets of $\{1, 2, \dots, n\}$ such that for any subset $S \subseteq \{1, 2, \dots, n\}$ of size $k$, the family $\{ A \cap S: A \in \mathcal{U}\}$ contains all $2^k$ subsets of $S$. 

\begin{thm}[\hspace{-1sp}\cite{Cygan2015Parameterized}]\label{thm:derandom_universal_set}
    For any positive integers $n$ and~$k$ we can construct an $(n,k)$-universal set $\mathcal{U}$ of size $|\mathcal{U}|=2^k\cdot k^{\mathcal{O}(\log k)}\cdot \log n$ in time $\mathcal{O}(n \cdot |\mathcal{U}|)$.
\end{thm}

\paragraph{Treewidth.} Intuitively, treewidth measures how similar a graph is to a tree. For an undirected graph $G$, a \emph{tree decomposition} is a pair~$(\mathbb{T},\mathcal B)$ where $\mathbb{T}$ is a tree and $\mathcal B$ is a collection of \emph{bags} in which each node $x$ of $\mathbb{T}$ has an assigned set of vertices $B_x \subseteq V(G)$. The bags have the property that for any edge of $G$, there is a bag in $\mathcal B$ that contains both of its endpoints. 
Furthermore, for each vertex $v$ of $G$, the bags containing $v$ form a subtree of $\mathbb{T}$.
The \emph{width} of $(\mathbb{T},\mathcal B)$ is defined as the largest bag size in $\mathcal B$ minus one.
The \emph{treewidth} of $G$ is the minimum width  over all tree decompositions of~$G$. 

Korhonen~\cite{korhonen2022treewidth} gives a 2-approximation for treewidth in the following result that has a better running time than the known exact algorithms computing treewidth (e.g., Bodlaender~\cite{bodlaender-linear-treewidth} with running time $2^{\mathcal{O}(k^3)}n$).
\begin{thm}[\hspace{-1sp}\cite{korhonen2022treewidth}]\label{thm:construct-treewidth}
     Given a graph $G$ on $n$ vertices and an integer $k$, there is an algorithm that outputs a tree decomposition of~$G$ of width at most $2k + 1$ or determines that the treewidth of~$G$ is larger than~$k$, in $2^{\mathcal{O}(k)} n$ time.
\end{thm}

\paragraph{Monadic second-order logic on graphs.} 
Monadic second-order logic can be used to describe graph properties. A \emph{first-order formula on graphs} is composed of atomic formulas ``$x = y$'' and ``$\texttt{Adj}(x,y)$'' over  variables using logical connectives $\wedge$ (\textsf{and}), $\vee$ (\textsf{or}) and $\neg$ (\textsf{negation}) and logical quantifiers $\forall$ and $\exists$. 
 Every formula defines a graph property by  interpreting variables as vertices, and the relation $\texttt{Adj} (x,y)$  as the vertices corresponding to~$x$ and~$y$ being adjacent. 
To facilitate the reading, we also make use of the usual symbols $\Rightarrow$, $\Leftrightarrow$, and $\nexists$.

A wider class of graph properties can be described by monadic second-order logic.
A \emph{monadic second-order formula} has four types of variables: variables for vertices and edges (denoted by lower-case letters) 
and variables for sets of vertices and sets of edges (denoted by upper-case letters).
The atomic formulas here also contain the $\texttt{Inc}(e,v)$ relation, interpreted as edge~$e$ being incident to vertex~$v$, 
and the relation of inclusion, denoted as ``$x \in X$''.  For example, the graph property of being bipartite is expressed by the formula \texttt{Bipartite} defined as
\begin{align*}
    \texttt{Bipartite} \equiv \exists V_1, V_2: {} & \bigl(\forall v:(v \in V_1 \vee v \in V_2 ) \wedge \neg (v \in V_1 \wedge v \in V_2)\bigr)   \\
    &{} \wedge \Bigl(\forall v_1,v_2:  \texttt{Adj}(v_1,v_2) \Rightarrow \bigl( (v_1 \in V_1 \wedge v_2 \in V_2) \vee (v_1 \in V_2 \wedge v_2 \in V_1) \bigr) \Bigr), 
\end{align*}
where $v$, $v_1$, and $v_2$ are vertex variables, while $V_1$ and $V_2$ are vertex set variables. 
We usually denote a monadic second-order formula by~$\varphi$.
We follow the notation by Satzinger~\cite{Satzinger2014Diploma}, using abbreviations such as $\exists v_1,v_2 : \varphi(v_1,v_2)$ instead of $\exists v_1 \exists v_2 \varphi(v_1,v_2)$, 
and $\nexists v \in V: \varphi(v)$ instead of 
$\nexists v(v \in V \wedge \varphi(v))$.
Although each variable in $\varphi$ has an associated type, we will omit to provide these explicitly whenever they are clear from the context. 

Note that whenever we need to decide whether a given graph~$G$ satisfies $\varphi$, we also need to provide an interpretation of every free variable in~$\varphi$; this way we can use certain vertices or edges in~$G$ as constants.
For instance, the following formula expresses whether a fixed vertex~$s$ is contained in a triangle of~$G$:
\[
\texttt{Triangle}_s \equiv
\exists v_1, v_2: \texttt{Adj}(s,v_1) \wedge 
\texttt{Adj}(s,v_2) \wedge
\texttt{Adj}(v_1,v_2)
\]
For further definitions and examples, see Lovász~\cite{lovasz2012large} and Cygan et al.~\cite {Cygan2015Parameterized}.

Courcelle proved in~\cite{courcelle1990monadic} that every graph property definable in monadic second-order logic can be decided in linear time on graphs of bounded treewidth. 
This metatheorem is not enough to solve the \SOP{} problem on graphs of bounded treewidth, but the extension of monadic second-order logic developed by Arnborg, Lagergren and Seese~\cite{ALS91} enables us to deal with edge weights.
Arnborg et al. defined \emph{linear EMS extremum problems}, which include problems where the input graph~$G=(V,E)$ is equipped with an evaluation function~$f$, and the task is to find the maximum (or the minimum) of the evaluation term $f(P)\coloneqq \sum_{e \in P} f(e)$ over all  sets~$P \subseteq E$ such that $G$ satisfies a given monadic-second order formula~$\varphi(P)$ with~$P$ as a free edge set variable in~$\varphi$.\footnote{We remark that such problems constitute only a small subset of linear EMS extremum problems; for the precise definition, see the paper by Arnborg et al.~\cite{ALS91} or the introductory description by Satzinger~\cite[Section 3.3]{Satzinger2014Diploma}.}

\begin{thm}[\hspace{-1sp}\cite{ALS91}] \label{thm:monadic_treewidth}   
    For each integer~$k$, every linear EMS extremum problem can be solved in linear time, assuming that a tree-decomposition of width at most~$k$ is provided for the input graph.
\end{thm}

\section{Initial observations}
\label{sec:tree-init}
Let $G=(V,E)$ be an undirected graph with a conservative weight function~$w\colon E \rightarrow \mathbb{R}$. Let $E^-=\{e \in E:w(e)<0\}$ denote the set of negative edges, and $\T$ be the set of negative trees they form. More precisely, let~$\T$ be the set of connected components in the subgraph of~$G$ spanned by~$E^-$, i.e., $\T$ contains those maximal, non-empty subsets of~$E^-$ that are connected in~$G$; the acyclicity of each $T \in \T$ follows from the conservativeness of~$w$.

\begin{lem}
\label{lem:closed-walk}
If $W$ is a closed walk that does not contain any edge  
of $E^-$ more than once, then $w(W) \geq w(E(W)) \geq 0$.
\end{lem}
\begin{proof}
Let $Z$ be the set of edges appearing an odd number of times on~$W$. Since any edge used at least twice in~$W$ has non-negative weight, we have $w(Z) \leq w(E(W))\leq w(W)$. Moreover, since $W$ is a closed walk, each vertex $v \in V$ has an even degree in the graph $(V,Z)$. 
Therefore, $(V,Z)$ is the union of edge-disjoint cycles, and the conservativeness of~$w$ implies $w(Z) \geq 0$.
\end{proof}

The following definition plays a crucial role in our algorithms.

\begin{defn}
    For a tree~$T \in \T$, a $T$-\emph{leap} $L$ with endpoints $a$ and $b$ is a path such that  $V(L) \cap V(T) = \{ a,b \}$ and $T \cap L=\emptyset$. A \emph{leap} is simply a $T$-leap for some~$T \in \T$.
    For a path $Q$, a \emph{leap on~$Q$} is a subpath of~$Q$ that is a leap.
    Furthermore, given a $T$-leap~$L$ with endpoints $a$ and $b$, we refer to the cycle $C = L \cup T[a,b]$ as the \emph{cycle induced by~$L$}. 
    We say that $L$ is \emph{parity-changing} if $|C|$ is odd. Finally, the \emph{shadow} of a $T$-leap $L$ on $Q$ with endpoints $a$ and $b$ is the edge set $Q\cap T[a,b]$, see Figure~\ref{fig:leap} for an illustration.
\end{defn}

\begin{figure}
    \centering
    \begin{tikzpicture}[scale = .7, nodeStyle/.style={circle,fill,minimum size=1pt}]
        \node (x) at (2,4) [circle, fill,label=left:\large{$x$}] {};
        \node (y) at (2,0) [nodeStyle,label=above:\large{$y$}] {};
        \node (a) at (3,1) [nodeStyle, label=below:\large{$a$}] {};
        \node (b) at (13,1) [nodeStyle, label=below:\large{$b$}] {};
    	
        \draw (a) edge[decorate,decoration=snake,line width=2.5pt,bend right=40,MyOrange] node[above]{ \large{$L$}}  (b);
        \draw[line width=1pt] (1,5) -- (4,2) -- (12,2) -- (15,5);
        \draw[line width=1pt] (1,-1) -- (4,2);
        \draw[line width=1pt] (12,2) -- (15,-1);
        \draw[line width=1pt, preaction={ draw,greenish!50,-, double=greenish!50, double distance=4\pgflinewidth}] (6,2) to (10,2);
        \draw[line width=1pt, decorate,decoration=snake] (x) ..controls (5,3).. (6,2) -- (10,2)..controls (14,4).. (b);
        \draw[line width=1pt, decorate,decoration=snake] (a) -- (y);
    \end{tikzpicture}  
    \caption{Illustration of a leap. Tree $T$ is represented by straight lines, while the $(x,y)$-path $Q$ is shown by curvy lines. The thick orange part of $Q$ forms a $T$-leap denoted by $L$. The shadow of $L$ is highlighted with green.}
    \label{fig:leap}
\end{figure}
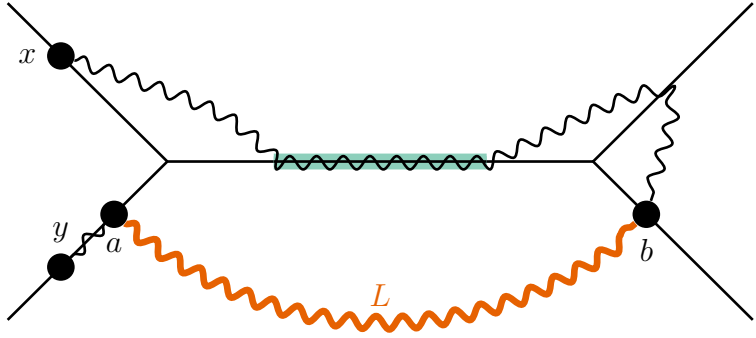

\begin{rem}\label{rem:min_parity-changing_leap}[Computing parity-changing leaps.]
    Suppose $\T$ contains a unique tree $T$. Then, given two vertices~$a$ and $b$ on~$T$, it is not difficult to obtain a parity-changing leap with endpoints~$a$ and~$b$ having minimum-weight: 
    we simply have to search for a minimum-weight $(a,b)$-path in $G-T$ whose parity differs from the parity of $T[a,b]$. 
    As $w$ restricted to $G-T$ is non-negative, we can apply Theorem~\ref{thm:SOP_nonnegative_weight} to find such a path.
\end{rem}

Now we describe a procedure that  redistributes the weights on the edges of~$G$ in such a way that the weight of each edge in the shadow of some leap $L$ becomes zero; this procedure will be useful for proving certain properties of the weight function~$w$ (namely, for proving Lemma~\ref{lem:T-min-pathlength}). 
See Algorithm~\ref{alg:initial-procedure} for a pseudocode. 
The input of the procedure is a set~$Q \subseteq E$ that can be partitioned into a set $\mathcal{P}$ of mutually vertex-disjoint paths, each having both endpoints on~$T$, 
and its output is a new weight function~$w_Q$ on $E$ fulfilling $w(Q)=w_Q(Q)$. 
Since $w$ is conservative, $w_Q(L) \geq 0$ remains true after performing lines~\ref{line:new_weight_forleap}--\ref{line:new_weight2} for some leap~$L$ on~$Q$.
After executing lines~\ref{line:new_weight_forleap}--\ref{line:new_weight2} for each $T$-leap on~$Q$, 
the resulting weight function~$w_Q$ fulfills the following: 
\begin{itemize}
    \item[(i)]  $w_Q(Q) = w(Q)$, so the total weight of~$Q$ remains unchanged,
    \item[(ii)] $w_Q(f)=0$ for each edge~$f$ in the shadow of some leap on~$Q$,
    and 
    \item[(iii)] $w_Q(L) \geq 0$ for each $T$-leap~$L$ on~$Q$.
\end{itemize}
In fact, Line~\ref{line:new_weight2} of Algorithm~\ref{alg:initial-procedure} can be defined differently, as long as the total weight on~$L$ decreases by $|w(f)|$.

\begin{algorithm}[H]
	\caption{Redistributing the weights along $Q$.}
	\label{alg:initial-procedure}
	\begin{algorithmic}[1]
        \Require{Graph $G$, set~$Q \subseteq E$ that can be partitioned into a set $\mathcal{P}$ of mutually vertex-disjoint 
        paths each having both endpoints on~$T$, weight function $w\colon E \to \mathbb{R}$.}
        \Ensure{A weight function~$w_Q\colon E \to \mathbb{R}$ such that $w(Q)=w_Q(Q)$.}
		\State Set $w_Q \equiv w$.
		\For{each $T$-leap $L$ on $Q$} 
		\For{ each edge $f$ in the shadow of $L$ for which $w_Q(f)=w(f)$} \label{line:new_weight_forleap}
		\State $w_Q(f)\coloneqq 0$. \label{line:new_weight1}
        \For{ each edge $e$ on the $T$-leap $L$}		
            \State $w_Q(e)\coloneqq w_Q(e)+w(f)/|L|$. 
            \label{line:new_weight2}
        \EndFor
        \EndFor
		\EndFor
        \State {\bf return} $w_Q$
	\end{algorithmic}
\end{algorithm}

The described procedure is essential for the following lemma.

\begin{lem}
\label{lem:T-min-pathlength}
Let $x,y,x',y'$ be four distinct vertices on a tree~$T$ in~$\T$. \begin{itemize}
    \item[(1)] If $Q$ is an $(x,y)$-walk in~$G$ that does not use any edge with negative weight more than once, then $w(Q) \geq w(T[x,y])$.
    \item[(2)] If $Q$ is an $(x,y)$-path, $Q'$ is an $(x',y')$-path, and $Q$ and $Q'$ are vertex-disjoint,  then \\
    $w(Q)+w(Q')\geq w(T[x,y] \triangle T[x',y'])$. 
\end{itemize} 
\end{lem}

\begin{proof}
To show statement~(1) of the lemma, let $Q$ be an $(x,y)$-walk in~$G$. Observe that we can assume w.l.o.g.\ that $Q$ is an $(x,y)$-path: if $Q$ contains cycles, then we can repeatedly delete any cycle from~$Q$, possibly of length~2, so that in the end we obtain an $(x,y)$-path whose weight is at most the weight of~$Q$, since $w$ is conservative.
Observe that any edge of~$Q \cap T$ that is not in the shadow of any $T$-leap on~$Q$ must lie on~$T[x,y]$: 
indeed, for any edge $uv \in Q \cap T$ that lies outside~$T[x,y]$, with $v$ being closer to $T[x,y]$ than~$u$ in~$T$, 
either $Q[u,x]$ or $Q[u,y]$ needs to use  a $T$-leap that contains~$uv$ in its shadow. 
Let $w_Q$ be the weight function obtained by Algorithm~\ref{alg:initial-procedure}. We get
\begin{align*}
    w(Q) = w_Q(Q) & = \sum_{f \in Q \cap T[x,y]} w_Q(f) +
    \sum_{f \in Q \cap T \setminus T[x,y]} w_Q(f) +
    \sum_{f \in Q \setminus T} w_Q(f) \\[2pt]
    & \geq w(Q \cap T[x,y])+0+\sum_{\substack{\text{$L$: $L$ is a} \\ \text{$T$-leap on~$Q$}}} w_Q(L) \\ & \geq  w(T[x,y]).
\end{align*}

To prove statement~(2), assume that $Q$ and $Q'$ are as in the statement.
Observe that any edge $f$ in~$(Q \cup Q') \cap T$ that is not contained in~$T[x,y] \triangle T[x',y']$ must be in the shadow of some $T$-leap in~$Q \cup Q'$. Indeed, if $f$ is not in $T[x,y] \cup T[x',y']$, then this follows by the same arguments we used for statement~(1). If $f \in T[x,y] \cap T[x',y']$, then 
one of the following cases holds because $f \in (Q \cup Q') \cap T$ by definition and  $Q$ and $Q'$ are vertex-disjoint:
either $f \in Q$ in which case it must be in the shadow of a $T$-leap on~$Q'$, or $f \in Q'$ in which case it is in the shadow of a $T$-leap on~$Q$. 
Considering the function~$w_{Q \cup Q'}$ obtained by the procedure in Algorithm~\ref{alg:initial-procedure}, thus we get
\begin{align*}
w(Q \cup Q') \, & = \,w_{Q \cup Q'}(Q \cup Q') \\
    &= \sum_{\substack{f \in Q \cup Q' \text{ and} \\[2pt] f \in  T[x,y] \triangle T[x',y']}} w_Q(f) +
    \sum_{\substack{f \in (Q \cup Q') \cap T\text{ and} \\[2pt] f \notin T[x,y] \triangle T[x',y']}} w_Q(f) +
    \sum_{f \in Q \cup Q' \setminus T} w_Q(f) \\[2pt]
&\geq w \Big( (Q \cup Q') \cap (T[x,y] \triangle T[x',y']) \Big) +0+\sum_{\substack{\text{$L$: $L$ is a } \\[1pt] \text{$T$-leap in~$Q \cup Q'$ }}} w_{Q \cup Q'}(L) \\
\, &\geq \,   w(T[x,y] \triangle T[x',y']). \qedhere
\end{align*}
\end{proof}

\begin{lem}
\label{lem:T-1parity-changing-leap}
Let $Q$ be a minimum-weight odd $(s,t)$-path in~$G$ that contains as few leaps as possible (among all minimum-weight odd $(s,t)$-paths). 
Then either $Q$ contains no leaps, or it contains at least one parity-changing leap.
\end{lem}

\begin{proof}
We show that whenever an optimal solution contains a leap, then it contains at least one parity-changing leap.
Let~$Q$ be a minimum-weight odd $(s,t)$-path with a minimum number of leaps, and suppose that $Q$ contains at least one leap, but no parity-changing leap. 
Let $T$ be a tree in~$\T$ for which $Q$ contains a $T$-leap, and let $x$ and $y$ be the first and the last vertex, respectively, of~$Q$ on~$T$ when we traverse~$Q$ in an arbitrarily fixed direction. Clearly, $x \neq y$ as $Q$ contains a $T$-leap and is a path. 

Consider the walk $W$ that we obtain if we traverse $Q[x,y]$ in one direction with the modification that whenever we reach a $T$-leap~$L$, we  replace~$L$ by the path contained in~$T$ between the endpoints of~$L$. Since no leap on~$Q$ is parity-changing, $|W|$ has the same parity as $|Q[x,y]|$.
Furthermore, any walk in~$T$ between two fixed vertices $x$ and $y$ has the same parity as  $|Q[x,y]|$, thus we can replace all edges of $Q[x,y]$ by the path~$T[x,y]$ without changing the parity of~$Q$. 
Therefore, $Q \setminus Q[x,y] \cup T[x,y]$ is an odd $(s,t)$-path, and its weight is at most~$w(Q)$ by statement~(1) of Lemma~\ref{lem:T-min-pathlength}. Moreover, it contains fewer leaps than~$Q$, contradicting the definition of~$Q$.
\end{proof}

We remark that if $\mathcal{T}$ contains only a single tree, then even a stronger statement holds: a minimum-weight odd $(s,t)$-path either uses no leaps, or it uses \emph{exactly one} parity-changing leap. 
However, even in this special case, a minimum-weight odd $(s,t)$-path may need to use an arbitrarily large number of leaps; see Figure~\ref{fig:interlaced} for an example.
\begin{figure}[t]
    \centering
    \begin{tikzpicture} [xscale = .65, nodeStyle/.style={circle, draw, minimum size=2em}]
        \node (s) at (2,0) [nodeStyle] {\footnotesize$s$};
            \node (t) at (4,0) [nodeStyle] {\footnotesize$t$};
            \node (v1) at (2,1) [nodeStyle] {};
            \node (a1) at (0,2) [nodeStyle] {};
            \node (x) at (2,2) [nodeStyle] {\footnotesize$x$};
            \node (y) at (4,2) [nodeStyle] {\footnotesize$y$};
            \node (v2) at (6,2) [nodeStyle] {};
            \node (a4) at (8,2) [nodeStyle] {};
            \node (b1) at (10,2) [nodeStyle] {};
            \node (a2) at (12,2) [nodeStyle] {};
            \node (b4) at (14,2) [nodeStyle] {};
            \node (v3) at (16,2)[nodeStyle] {};       
            \node (a3) at (16,1) [nodeStyle] {};
             \node (a) at (18,2) [nodeStyle] {\footnotesize$a$};
            \node (b3) at (20,2) [nodeStyle] {};
            \node (b) at (22,2) [nodeStyle] {\footnotesize$b$};
    
        \foreach \a \b in {s/v1,v1/x, t/y}{
            \draw (\a) edge[line width=1pt]  (\b);};
        \foreach \a \b in {a1/x,x/y,y/v2,a4/v2,a4/b1,b1/a2,a2/b4,b4/v3,v3/a,a/b3,b3/b,a3/v3}{
            \draw (\a) edge[, line width=1pt, double, double distance=2pt, MyPurple]  (\b);};     
        \draw (a) edge[line width=2.5pt,bend left=40,MyOrange] node[above]{ \large{$L$}}  (b);
        \foreach \a \b in {a1/b1/,a2/a, b3/a3,b4/a4}{
            \draw (\a) edge[line width=1pt,bend left=30] (\b);};
    \end{tikzpicture}
    \caption{An example where there is a unique odd $(s,t)$-path~$P$. The negative tree $T^-$ is shown with purple double lines, and the unique parity-changing leap $L$ is shown with an orange thick line. The path~$P$ must use all leaps in the figure.} 
    \label{fig:interlaced}
\end{figure}
By contrast, a shortest $(s,t)$-path intersects each tree  $T \in \T$ in a single (possibly empty) path, so it never uses leaps. 
The next result is quite trivial, but we present its proof for completeness.
\begin{lem}
\label{lem:shortest-path-simple}
Let $P$ be a minimum-weight $(s,t)$-path in~$G$, 
and let $T$ be a tree in $\T$.
Then for any two vertices~$u$ and~$v$ in~$V(P) \cap V(T)$, it holds that $P[u,v]=T[u,v]$.
\end{lem}
\begin{proof}
If the statement of the lemma fails, then there exist vertices~$u$ and~$v$ in~$T$ such that the only vertices of~$P$ appearing on $T[u,v]$ are~$u$ and~$v$, and $uv$ is not an edge of~$P$. 
Then Lemma~\ref{lem:closed-walk} implies 
$w(P[u,v])+w(T[u,v])\geq 0$. 
Thus, replacing the subpath~$P[u,v]$ of~$P$ with~$T[u,v]$ results in an $(s,t)$-path whose weight is strictly less than $w(P)$, a contradiction.
\end{proof}

\section{Parity Constrained Odd Paths}
\label{sec:parity}

In this section, we consider the \SPaPrOP{} problem, a variant of \SOP{} in which some edges have parity constraints on their position along the path. 
Let us recall the following definition stated in Section~\ref{sec:intro}.

\begin{defn}\label{def:parity}
Let $F_\even \subseteq E$ and $F_\odd \subseteq E$ be two disjoint edge sets in~$G$.
An odd path $Q$ is \emph{$(F_\even,F_\odd)$-constrained} if the sequence number of each edge in $Q \cap F_\even$ is even, and the sequence number of each edge in $Q\cap F_\odd$ is odd. For an illustration, see Figure~\ref{fig:def_SPaPrOP}.
\end{defn}

\begin{figure}
    \centering
    \begin{subfigure}[t]{0.45\linewidth}
        \begin{tikzpicture}[xscale=.8,yscale=.9,nodeStyle/.style={circle, draw, minimum size=2em}]
           \node (aux) at (0,-2) {};
           \node (s) at (0,2) [nodeStyle] {\footnotesize$s$};
            \node (v1) at (1,0) [nodeStyle] {\footnotesize$v_1$};
            \node (v2) at (2,2) [nodeStyle] {\footnotesize$v_2$};
            \node (v3) at (3,0) [nodeStyle] {\footnotesize$v_3$};
            \node (v4) at (4,2) [nodeStyle] {\footnotesize$v_4$};
            \node (v5) at (6,2) [nodeStyle] {\footnotesize$v_5$};
            \node (t) at (5,0) [nodeStyle] {\footnotesize$t$};
        
            \foreach \a \b in {s/v1,v2/v3,v4/t}{
                \draw (\a) edge[line width=1.5pt, dashed, greenish]  (\b);};
            \foreach \a \b in {s/v2,s/v3,v1/v4,v4/v5}{
                \draw (\a) edge[line width=1pt ]  (\b);};
            \foreach \a \b in {v1/v2,v3/v4}{
                \draw (\a) edge[line width=1.5pt, densely dotted, MyOrange]  (\b);};  
        \end{tikzpicture}
        \caption{Illustration of Definition~\ref{def:parity}.}
        \label{fig:def_SPaPrOP}
    \end{subfigure}
    \begin{subfigure}[t]{0.45\linewidth}
        \begin{tikzpicture}[xscale=.8,yscale=.6,nodeStyle/.style={circle, draw, minimum size=2em}]
            \node (s) at (0,2+5) [nodeStyle] {\footnotesize$s$};
            \node (v1) at (1,0+5) [nodeStyle] {\footnotesize$v_1$};
            \node (v2) at (2,2+5) [nodeStyle] {\footnotesize$v_2$};
            \node (v3) at (3,0+5) [nodeStyle] {\footnotesize$v_3$};
            \node (v4) at (4,2+5) [nodeStyle] {\footnotesize$v_4$};
            \node (t) at (5,0+5) [nodeStyle] {\footnotesize$t$};
            \node (v5) at (6,2+5) [nodeStyle] {\footnotesize$v_5$};
            \node (s-2) at (0,2) [nodeStyle, gray!30] {\footnotesize$s'$};
            \node (v1-2) at (1,0) [nodeStyle] {\footnotesize$v'_1$};
            \node (v2-2) at (2,2) [nodeStyle] {\footnotesize$v'_2$};
            \node (v3-2) at (3,0) [nodeStyle] {\footnotesize$v'_3$};
            \node (v4-2) at (4,2) [nodeStyle] {\footnotesize$v'_4$};
            \node (t-2) at (5,0) [nodeStyle, gray!30] {\footnotesize$t'$};
            \node (v5-2) at (6,2) [nodeStyle] {\footnotesize$v'_5$};
            
            \foreach \a \b in {v1/v1-2,v2/v2-2,v3/v3-2,v4/v4-2,v5/v5-2}{ %
                \draw (\a) edge[line width=.7pt,dashdotted]  (\b);};
            \foreach \a \b in {s/s-2,t/t-2}{ %
                \draw (\a) edge[line width=.7pt,dashdotted, gray!30]  (\b);};
                
            \foreach \a \b in {s/v1,v2/v3,v4/t}{%
                \draw (\a) edge[line width=1.5pt, dashed, greenish]  (\b);};
            \foreach \a \b in {s/v2,s/v3,v1/v4,v4/v5}{
                \draw (\a) edge[line width=1pt ]  (\b);};
            \foreach \a \b in {v1/v2,v3/v4}{ %
                \draw (\a) edge[line width=1.5pt, densely dotted, gray!30]  (\b);};
            
            \foreach \a \b in {s-2/v1-2,v2-2/v3-2,v4-2/t-2}{%
                \draw (\a) edge[line width=1.5pt, dashed, gray!30]  (\b);};
            \foreach \a \b in {s-2/v2-2,s-2/v3-2}{%
                \draw (\a) edge[line width=1.5pt, gray!30]  (\b);};
            \foreach \a \b in {v1-2/v4-2,v4-2/v5-2}{
                \draw (\a) edge[line width=1pt ]  (\b);};
            \foreach \a \b in {v1-2/v2-2,v3-2/v4-2}{ %
                \draw (\a) edge[line width=1.5pt, densely dotted, MyOrange]  (\b);};              
        
            \path[draw, rounded corners] ($(s)+(-1,1)$) -- ($(v5)+(1,1)$) -- ($(v5)+(1,-2-1)$) -- ($(s)+(-1,-2-1)$) -- cycle;
            \node (G) at ($(s)+(-1.5,1)$) {$G$};
            \path[draw, rounded corners] ($(s-2)+(-1,1)$) -- ($(v5-2)+(1,1)$) -- ($(v5-2)+(1,-2-1)$) -- ($(s-2)+(-1,-2-1)$) -- cycle;
            \node (G-2) at ($(s-2)+(-1.5,1)$) {$G'$};
        \end{tikzpicture}
        \caption{Construction for Theorem~\ref{thm:parity-constr-oddpath-inP}.}
        \label{fig:SPaPrOP_construction}
    \end{subfigure}
        
    \caption{(a) Edges in $F_\even$ are shown as dotted orange lines, $F_\odd$ as dashed green lines, and $E \setminus (F_\even \cup F_\odd)$  as solid black. Note that there exists three  $(F_\even,F_\odd)$-constrained odd  $(s,t)$-paths, namely $Q_1 = s,v_1,v_2,v_3,v_4,t$, $Q_2=s,v_3,v_4t$ and $Q_3=s,v_1,v_4,t$.
    (b) Graph $H$ obtained from (a). Removed arcs and vertices are shown in gray. Dashdotted lines have weight 0, while the rest have original weight.}
    \label{fig:all_def_SPaPrOP}
\end{figure}
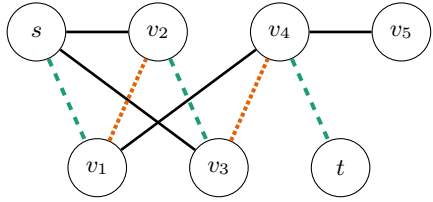
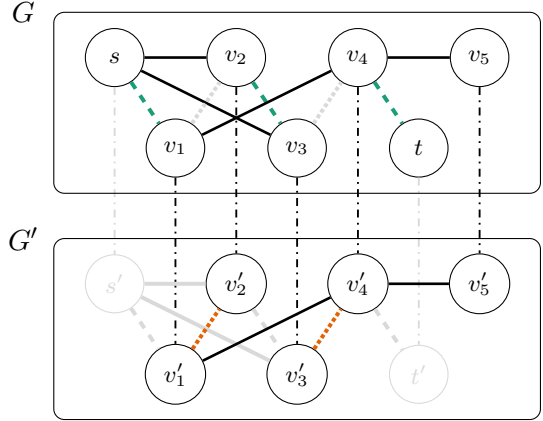

\begin{center}
\fbox{ 
\parbox{13.6cm}{
\begin{tabular}{l}\SPaPrOP{}:  \end{tabular} \\
\begin{tabular}{p{1cm}p{11.5cm}}
\textbf{Input}: & An undirected graph $G=(V,E)$, a weight function~$w\colon E \to \mathbb{R}$, two vertices $s$ and $t$ in~$G$, and disjoint edge sets $F_\even\subseteq E$ and $F_\odd\subseteq E$. \\
\textbf{Goal}: & Find a minimum-weight $(F_\even,F_\odd)$-constrained odd $(s,t)$-path.
\end{tabular}
}}
\end{center}

This problem is a key step for our algorithms solving the \SOP{} problem.
We show that \SPaPrOP{} can be solved in polynomial time if all negative edges are contained in $F_\even\cup F_\odd$, assuming conservative weights. 

\begin{thm}
\label{thm:parity-constr-oddpath-inP}
If $w$ is conservative and $F_\even \cup F_\odd$ contains all negative edges of the input graph $G=(V,E)$, then the \SPaPrOP{} problem can be solved in $\mathcal{O}(mn+n^2 \log n)$ time, where $n=|V|$ and $m=|E|$.
\end{thm}
\begin{proof}
We present a generalization of the well-known algorithm for determining a shortest odd (or, equivalently, even) path between two vertices in an undirected graph with non-negative edge weights. 

Let $(G,w,s,t,F_\even,F_\odd)$ be our input 
 for \SPaPrOP{} with $G=(V,E)$ such that $F_\even \cup F_\odd$ contains all edges with negative weight.
An illustration of the following construction is shown in Figure~\ref{fig:SPaPrOP_construction}.
We start by creating a copy~$G'$ of~$G$ where each vertex or edge $x$ in~$G$ has a copy~$x'$ in~$G'$. 
The copy of each edge~$e \in E$ has the same weight as~$e$, and additionally, we add a new edge between each vertex~$v \in V$ and its copy~$v'$ of weight~0. 
Finally we delete~$F_\even$ from~$G$, and we delete the edge set~$F_\odd':=\{e':e \in F_\odd\}$ as well as vertices~$s'$ and $t'$ from $G'$.
Let~$H$ denote the graph thus obtained. 

We claim that the minimum weight of an $(F_\even,F_\odd)$-constrained odd $(s,t)$-path in~$G$ equals the minimum weight of a perfect matching in~$H$. First, let $Q$ be an $(F_\even,F_\odd)$-constrained odd $(s,t)$-path. We define an edge set $M_Q$ as follows:
\begin{equation*}
    M_Q \coloneqq \{e: e \in Q,\ \sq(e,Q) \textrm{ is odd}\}
        \cup \{e': e \in Q,\ \sq(e,Q) \textrm{ is even}\}
        \cup \{v v': v \in V \setminus V(Q)\}.
\end{equation*}

Observe first that since~$Q$ is $(F_\even,F_\odd)$-constrained, each edge of $M_Q$ is present in~$H$. 
Therefore, it is easy to see that $M_Q$ is a perfect matching in~$H$, and its weight is exactly~$w(Q)$.

Second, let~$M$ be a perfect matching in~$H$. 
Let $M_1=\{e \in E: e \in M\}$ and $M_2=\{e \in E: e' \in M\}$.
Consider the multiset~$Q_M= M_1 \uplus M_2$ of edges.
Note that for each vertex~$v$ in~$V \setminus \{s,t\}$, either $v$ is incident to one edge of~$M_1$ and one edge of~$M_2$ (counting multiplicities), or $vv' \in M$.
Vertices~$s$ and~$t$ each are incident to one edge of~$M_1$ and no edge of~$M_2$. 
Moreover, an edge~$e \in E$ may be present both in~$M_1$ and $M_2$, but such an edge cannot have negative 
weight since 
$F_\even \cup F_\odd$ contains all edges with negative weight.

Therefore, each connected component of the subgraph~$(V,M_1 \cup M_2)$ of~$G$ is either an odd~$(s,t)$-path~$Q$ in~$G$ (alternating between~$M_1$ and~$M_2$), an even cycle in~$G$, or a single edge in $E \setminus (F_\even \cup F_\odd)$.
Recall that each cycle in~$G$ has non-negative weight, as $w$ is conservative.
Moreover, each $e \in E \setminus (F_\even \cup F_\odd)$ has non-negative weight by our assumption $E^- \subseteq F_\even \cup F_\odd$. 
Thus $w(Q) \leq w(M)$, and hence $Q$ is an odd $(s,t)$-path of weight at most~$w(M)$. 
Finally, observe that
 any edge of~$F_\even \cap Q$ belongs to~$M_2$ (and not~$M_1$), and  thus has an even sequence number in~$Q$, and similarly, 
any edge of~$F_\odd \cap Q$ belongs to~$M_1$ (and not~$M_2$), and  thus has an odd sequence number in~$Q$, by our construction of~$H$.
Therefore, $Q$ is indeed an $(F_\even,F_\odd)$-constrained odd $(s,t)$-path, and has weight at most~$w(M)$.
This proves our claim.

As a consequence, in order to find a minimum-weight  $(F_\even,F_\odd)$-constrained odd $(s,t)$-path, we simply have to construct~$H$, find a minimum-weight perfect matching in~$H$ using, e.g., Gabow's algorithm~\cite{Gabow90}, and build the corresponding odd $(s,t)$-path as explained above.
The overall running time is $\mathcal{O}(mn+n^2 \log n)$.
\end{proof}

\section{Negative edges forming a tree}
\label{sec:tree}

In this section, we consider the special case of \SOP{} when the weight function~$w$ is conservative, and the set of edges whose weight is negative forms a tree~$T^-$ in~$G$. We give an algorithm for solving the problem in polynomial time, presented as Algorithm~\ref{alg:SOP-tree}. An illustration can be seen in Figure~\ref{fig:SOP_algoTREE}.

The high-level idea of the algorithm is as follows. 
By Lemma~\ref{lem:T-1parity-changing-leap}, any shortest odd $(s,t)$-path either contains no leaps, or it contains a parity-changing leap, hence our algorithm checks both cases and it generates two types of paths. 
First, to find a minimum-weight odd $(s,t)$-path~$P$ among those that contain  no leaps, we guess the first and the last vertices of~$P$ that appear on~$T^-$ 
by picking two vertices $a$ and $b$. Then we partition the path $T^-[a,b]$ into two sets of alternating edges, 
and solve an instance of the  \SPaPrOP{} problem with this partition.
This gives us the first type of path.
Second, to find a minimum-weight odd $(s,t)$-path~$P$ among those that contain  a parity-changing leap~$L$,
we guess the endpoints of~$L$ by picking two vertices $a$ and $b$. Then we compute a minimum-weight parity-changing leap~$L'$ with endpoints~$a$ and~$b$, and construct an odd $(s,t)$-path by using~$L'$ and two vertex-disjoint paths from~$\{s,t\}$ to~$\{a,b\}$ of minimum weight. 
This gives us the second type of path.
The output of the algorithm is the best path among these two types of solutions.

\begin{rem} \label{rem:disjoint-paths}
    The task on line~\ref{algLine:vertex_disj_paths} is to solve an instance of the \DISP{} problem that, given an undirected graph with conservative edge weights and terminals~$s_1,s_2,t_1,$ and~$t_2$,  the goal is to find two vertex-disjoint paths leading from~$\{s_1,s_2\}$ to~$\{t_1,t_2\}$ with minimum total weight. 
    This problem was shown to be $\NP$-hard
     by Schlotter and Sebő~\cite{schlotter2022shortest} even with a conservative weight function of the form $w\colon E \to \{ -1,1\} $. 
    Nevertheless, if the negative-weight edges form a tree, as in our case, then the problem can be solved in polynomial time~\cite{schlotter2023twopaths}.
\end{rem}

The algorithm uses the following definition.
\begin{defn}
    Given two distinct vertices~$a$ and~$b$ on a cycle~$C$, we say that they \emph{cut~$C$ into paths} $C_1$ and $C_2$ if both $C_1$ and $C_2$ have endpoints $a$ and $b$, and $C_1 \cup C_2 = C$.
\end{defn}

We first prove the correctness of Algorithm~\ref{alg:SOP-tree}.

\begin{algorithm}[t]
\caption{Solving \SOP{} with conservative weight function and the set of negative edges forms a tree.
}
\label{alg:SOP-tree}
\begin{algorithmic}[1]
\Require{A graph $G=(V,E)$, a conservative weight function $w\colon E \to \mathbb{R}$, vertices $s,t \in V$,  a set of negative edges $T^-=\{e \in E, w(e)<0\}$ forming a tree. }
\Ensure{A minimum-weight odd $(s,t)$-path or $\varnothing$ 
if no such path exists.}
\State Set $S \coloneqq \emptyset$, $w_S\coloneqq \infty$.
\ForAll{distinct vertices $a$ and $b$ on $T^-$}  \label{algLine:for_loop} \vspace{2pt}\Comment{\textit{Computations for first type of paths}}
	\State Set $F^a_\even \coloneqq \{e \in T^-[a,b]: \sq(e,T^-[a,b],a) \equiv 0 \mod 2\}$.\label{algLine:F_even} 
    
	\State Set $F^a_\odd \coloneqq \{e \in T^-[a,b]: \sq(e,T^-[a,b],a) \equiv 1 \mod 2\}$.\label{algLine:F_odd}
	\State Set $G'\coloneqq G-(T^- \setminus T^-[a,b])$. 
    \State Compute $Q \coloneqq \SPaPrOP (G',w,s,t,F^a_\even,F^a_\odd)$.\label{algLine:inst1}
	\If{$w(Q)<w_S$} set $S\coloneqq Q$, $w_S\coloneqq w(S)$. 	\EndIf
    \label{algLine:def_path2odd}
	
    \State Compute $Q \coloneqq \SPaPrOP (G',w,s,t,F^a_\odd,F^a_\even)$.\label{algLine:inst2}
	\If{$w(Q)<w_S$} set $S\coloneqq Q$, $w_S\coloneqq w(S)$. \label{algLine:def_path2even}	\EndIf 
 \vspace{3pt}
    \If{$\exists$ parity-changing leap from~$a$ to~$b$}    \Comment{\textit{Computations for second type of paths}} 
    \vspace{2pt}
\State Compute a minimum-weight parity-changing leap $L$ with endpoints~$a$ and $b$ in~$G$. \label{algLine:leap} 
 
 \Comment{{\it See Remark~\ref{rem:min_parity-changing_leap}}.}
    \If{$\exists$ two vertex-disjoint paths from~$\{s,t\}$ to~$\{a,b\}$}    
	\parState{Compute two vertex-disjoint paths~$P_s$ and~$P_t$ from $\{s,t\}$ to $\{a,b\}$ in~$G$ that minimizes $w(P_s)+w(P_t)$, with $s \in V(P_s)$ and $t \in V(P_t)$. \Comment{{\it See Remark~\ref{rem:disjoint-paths}}.}%
 \label{algLine:vertex_disj_paths}}    	\State Set $C\coloneqq L \cup T^-[a,b]$.
	\State Let $x$ be the vertex of~$C$ closest to~$s$ on~$P_s$, and set $P'_s \coloneqq P_s[s,x]$.\label{algLine:def_Ps}
	\State Let $y$ be the vertex of~$C$ closest to~$t$ on~$P_t$, and set $P'_t \coloneqq P_t[y,t]$.
	\State Let $C_1$ and $C_2$ be the two paths into which $x$ and $y$ cut $C$.
	\State Set $S_1\coloneqq P'_s \cup C_1 \cup P'_t$ and 
	$S_2\coloneqq P'_s \cup C_2 \cup P'_t$.
	\If{$S_1$ is odd}{ $S^\star\coloneqq S_1$
  {\bf else} $S^\star\coloneqq S_2$.}
	\EndIf
	\If{$w(S^\star)<w_S$} set $S\coloneqq S^\star$, $w_S\coloneqq w(S^\star)$.\label{algLine:def_path3} \EndIf
   \EndIf
   \EndIf
	\EndFor
\If{$S \neq \emptyset$} {\bf return}  $S$;
\Else{}  {\bf return } $\varnothing$.
\EndIf 
\end{algorithmic}
\end{algorithm}

\begin{figure}[ht]
    \centering
    \begin{subfigure}[t]{.45\textwidth}
        \begin{tikzpicture}[nodeStyle/.style={draw,circle,minimum size=2em}]
            \node (s) at (0,0) [nodeStyle] {\footnotesize$s$};
            \node (t) at (6,0) [nodeStyle] {\footnotesize$t$};
            \node (v1) at (0,3) [nodeStyle] {\footnotesize$v_1$};
            \node (v2) at (2,3) [nodeStyle] {\footnotesize$v_2$};
            \node (v3) at (4,3) [nodeStyle] {\footnotesize$v_3$};
            \node (v4) at (6,3) [nodeStyle] {\footnotesize$v_4$};
            \node (v5) at (4,1.5) [nodeStyle] {\footnotesize$v_5$};
            \node (v6) at (4,0) [nodeStyle] {\footnotesize$v_6$};
            \node (v7) at (2,0) [nodeStyle] {\footnotesize$v_7$};     
            \foreach \a \b \wpos in {v1/v2/above,v2/v3/above,v3/v4/above,v3/v5/left,v5/v6/left}{
                \draw (\a) edge[double,double distance=2pt,line width=1pt,MyPurple] node[\wpos] {\footnotesize -1} (\b);};
            \foreach \a \b \wpos \w  in {v1/v7/below left/2,s/v7/below/1,v7/v6/below/2,v6/t/below/3,v4/v5/below/3, v4/t/right/0, v2/v7/left/1}{
                \draw (\a) edge[line width=.6pt] node[\wpos] {\footnotesize \w}  (\b);}; 
        \end{tikzpicture}
        \caption{Original instance. }
        \label{fig:SOP_algoTREE-A}
    \end{subfigure}
    \begin{subfigure}[t]{.45\textwidth}
        \begin{tikzpicture}[nodeStyle/.style={draw,circle,minimum size=2em}]
            \node (s) at (0,0) [nodeStyle] {\footnotesize$s$};
            \node (t) at (6,0) [nodeStyle] {\footnotesize$t$};
            \node (v1) at (0,3) [nodeStyle] {\footnotesize$v_1$};
            \node (v2) at (2,3) [nodeStyle] {\footnotesize$a$};
            \node (v3) at (4,3) [nodeStyle] {\footnotesize$v_3$};
            \node (v4) at (6,3) [nodeStyle] {\footnotesize$v_4$};
            \node (v5) at (4,1.5) [nodeStyle] {\footnotesize$v_5$};
            \node (v6) at (4,0) [nodeStyle] {\footnotesize$b$};
            \node (v7) at (2,0) [nodeStyle] {\footnotesize$v_7$};    
            \foreach \a \b \wpos in {v1/v2/above,v3/v4/above}{
                \draw (\a) edge[line width=2pt,gray!30 ] node[\wpos] {\footnotesize -1} (\b);};
            \foreach \a \b \wpos   in {v2/v3/above,v5/v6/left}{%
                \draw (\a) edge[line width=1.8pt, dashed, greenish]  node[\wpos] {\footnotesize -1}  (\b);};
            \foreach \a \b \wpos  in {v3/v5/left}{ %
                \draw (\a) edge[line width=1.8pt, densely dotted, MyOrange]  node[\wpos] {\footnotesize -1}  (\b);};
            \foreach \a \b \wpos \w  in {v1/v7/below left/2, v7/v2/left/1,v6/t/below/$\phantom{,,}$3}{
                \draw (\a) edge[line width=.6pt] node[\wpos] {\footnotesize \w}  (\b);};
            \foreach \a \b \wpos \w  in {v7/s/below/2,v6/v7/below/1,v4/v5/below/3,v4/t/right/0}{
                \draw (\a) edge[line width=.6pt, decorate, decoration=zigzag] node[\wpos] {\footnotesize \w}  (\b);};
            \draw (v5.south east) edge[line width=1pt, decorate, decoration=zigzag,greenish] (v6.north east);
        \end{tikzpicture}
    \caption{Odd paths computed on lines~\ref{algLine:F_even}--\ref{algLine:def_path2even}.}
    \label{fig:SOP_algoTREE-C}
    \end{subfigure}
    \begin{subfigure}[t]{.45\textwidth}
        \begin{tikzpicture}[nodeStyle/.style={draw,circle,minimum size=2em}]
            \node (s) at (0,0) [nodeStyle] {\footnotesize$s$};
            \node (t) at (6,0) [nodeStyle] {\footnotesize$t$};
            \node (v1) at (0,3) [nodeStyle] {\footnotesize$v_1$};
            \node (v2) at (2,3) [nodeStyle] {\footnotesize$a$};
            \node (v3) at (4,3) [nodeStyle] {\footnotesize$v_3$};
            \node (v4) at (6,3) [nodeStyle] {\footnotesize$v_4$};
            \node (v5) at (4,1.5) [nodeStyle] {\footnotesize$v_5$};
            \node (v6) at (4,0) [nodeStyle] {\footnotesize$b$};
            \node (v7) at (2,0) [nodeStyle] {\footnotesize$v_7$};
            \foreach \a \b \wpos \w  in {v7/v6/below/2,v2/v7/right/1}{
                \draw (\a) edge[line width=1pt,MyOrange] node[\wpos] {\footnotesize \w}  (\b);};
            \path (v7) +(-0.35,+0.15) coordinate (v7');
            \foreach \a \b \wpos \w  in {v4/v5/below/3,  v6/t/below/3,v1/v2/above/-1,v1/v7'/below left/2}{
                \draw (\a) edge[line width=.6pt] node[\wpos] {\footnotesize \w}  (\b);};
            \foreach \a \b \wpos in {v3/v5/left,v5/v6/left}{
                \draw (\a) edge[line width=0.8pt, preaction={draw,greenish!50,-, double=greenish!50, double distance=2\pgflinewidth}] node[\wpos] {\footnotesize -1} (\b);};   
            \draw (v7) edge[line width=.6pt, decorate, decoration=zigzag, preaction={draw,MyPurple!50,-, double=MyPurple!50, double distance=2\pgflinewidth}]  node[below] {\footnotesize 1} (s);
            \draw (v7.north west)  edge[line width=.6pt, decorate, decoration=zigzag, preaction={draw,MyPurple!50,-, double=MyPurple!50, double distance=2\pgflinewidth}](v2.south west);
            \draw (v2) edge[line width=0.6pt, decorate, decoration=zigzag] node[above] {\footnotesize -1}  (v3);
            \draw (v3) 
            edge[line width=.6pt, decorate, decoration=zigzag, preaction={draw,greenish!50,-, double=greenish!50, double distance=2\pgflinewidth}] 
            node[above] {\footnotesize -1} (v4);
            \draw (v4) edge[line width=.6pt, decorate, decoration=zigzag, preaction={draw,greenish!50,-, double=greenish!50, double distance=2\pgflinewidth}] node[right, xshift=3pt] {\footnotesize 0}  (t);
        \end{tikzpicture}
        \caption{Odd path computed on lines~\ref{algLine:leap}--\ref{algLine:def_path3}.}
        \label{fig:SOP_algoTREE-D}
    \end{subfigure}
    \caption{Example of \SOP{} instance with conservative weights when running Algorithm~\ref{alg:SOP-tree} for $a=v_2$ and $b=v_6$. With these specific $a$ and $b$, the solution generated up to this point is given in Figure~(c). \\
    (a) Original instance. Tree $T^-$ is shown using purple double lines. \\
    (b) Paths obtained on lines~\ref{algLine:F_even}--\ref{algLine:def_path2even}. 
    Deleted edges are shown in gray. Set $F^a_\even$ is shown by dotted orange lines, and $F^a_\odd$ by dashed green lines. Solving the \SPaPrOP{} problem with instance
    $(G',w,s,t,F^a_\odd,F^a_\even)$ generates the path on vertices~$s,v_7,b,t$ which is a shortest odd $(s,t)$-path in $G-T^-$. Solving the instance $(G',w,s,t,F^a_\even,F^a_\odd)$ generates the zigzagged path.\\
    (c) Path obtained on lines~\ref{algLine:leap}--\ref{algLine:def_path3} shown in zigzag. The parity-changing leap $L=\{ (a,v_7), (v_7,b) \}$ is shown in orange. An $(s,a)$-path of minimum weight is highlighted in purple and an $(t,b)$-path of minimum weight is highlighted in green.}
    \label{fig:SOP_algoTREE}
\end{figure}
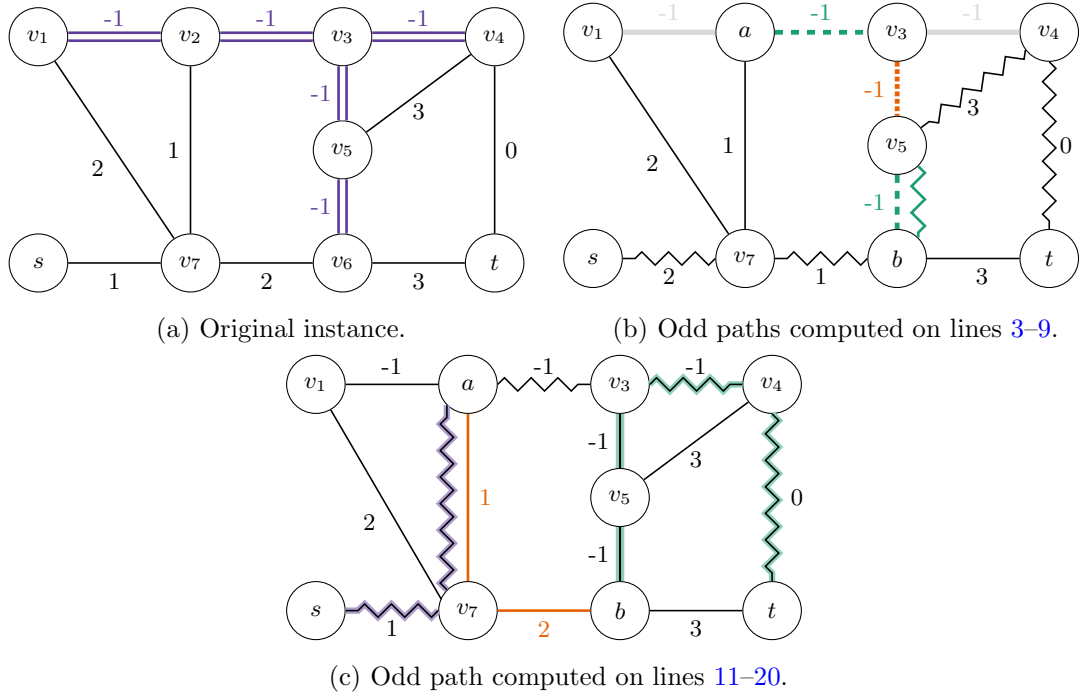

\begin{lem}
\label{lem:algoTREE-correctness}
Algorithm~\ref{alg:SOP-tree} correctly solves the \SOP{} problem when the weight function is conservative and negative edges form a tree.
\end{lem}
\begin{proof}
First we show that any path that the algorithm sets as $S$ is an odd $(s,t)$-path. This is trivial for the paths on lines~\ref{algLine:def_path2odd} and~\ref{algLine:def_path2even}. Consider now the path~$S^\star$ on line~\ref{algLine:def_path3}. 
Clearly, the paths~$P'_s$ and~$P'_t$ are vertex-disjoint, since $P_s$ and~$P_t$ are vertex-disjoint by definition. 
Second, the vertex shared by $P'_s$ and~$C$ is~$x$, by the definition of~$x$, so the paths $P'_s$, $C_1$, and~$C_2$ are pairwise openly disjoint.
Similarly, by our choice of~$y$, the paths $P'_t$, $C_1$, and~$C_2$ are pairwise openly disjoint. Hence, $S_i=P'_s \cup C_i \cup P'_t$ is a path for both~$i=1,2$. Moreover, since $L$ is a parity-changing leap, we know that $|C_1|+|C_2|$ is odd. Therefore, exactly one of~$S_1$ and~$S_2$ is an odd path, and the algorithm considers it as~$S^\star$ on line~\ref{algLine:def_path3}.

Thus, to prove the correctness of our algorithm, it suffices to show that if there exists an odd $(s,t)$-path of weight~$w^\star$, then our algorithm picks an odd $(s,t)$-path of weight at most~$w^\star$. 

Let~$Q^\star$ be a shortest odd $(s,t)$-path with weight $w(Q^\star) = w^\star$ containing as few leaps as possible. By Lemma~\ref{lem:T-1parity-changing-leap}, 
either $Q^\star$ contains no leaps, or it contains a parity-changing leap.
See an illustration of the algorithm on Figure~\ref{fig:SOP_algoTREE}.

First assume that $Q^\star$ contains no leaps. Then, by Lemma~\ref{lem:shortest-path-simple} either (i) $Q^\star \cap T^- = \emptyset$, or (ii) $Q^\star \cap T^-$ is a path containing at least one edge. 
In case~(i), observe that $Q^\star$ is a path in $G-T^-$. Hence, it is an $(F_1,F_2)$-constrained path for every pair~$(F_1,F_2)$ of subsets of~$T^-$. In fact, for every choice of~$a$ and~$b$, 
$Q^\star$ is a solution for the two \SPaPrOP{} instances constructed on lines~\ref{algLine:inst1} and~\ref{algLine:inst2}. Consequently,
the algorithm will find an odd $(s,t)$-path with weight at most~$w^\star$ on both lines~\ref{algLine:inst1} and~\ref{algLine:inst2},
and will select the one with minimum weight in either line~\ref{algLine:def_path2odd}  or line~\ref{algLine:def_path2even}.
In case~(ii), let $a$ and $b$ be the endpoints of the path $Q^\star \cap T^-$, and consider the iteration of Algorithm~\ref{alg:SOP-tree} when the algorithm picks vertices~$a$ and~$b$ in the for-loop on line~\ref{algLine:for_loop}. Let us define 
$F^a_\even$ and $F^a_\odd$ as on lines~\ref{algLine:F_even} and~\ref{algLine:F_odd} of Algorithm~\ref{alg:SOP-tree}. 
Clearly, $Q^\star$ is either $(F^a_\even,F^a_\odd)$-constrained, or $(F^a_\odd,F^a_\even)$-constrained for a pair of vertices $a,b$ in $T$, moreover, it is a path in~$G'$ obtained by deleting all edges of~$T$ from~$G$ except those of $T^-[a,b]$. 
Hence, Algorithm~\ref{alg:SOP-tree} stores a path of weight at most~$w(Q^\star)=w^\star$ either on line~\ref{algLine:def_path2odd} or on line~\ref{algLine:def_path2even}.
This shows the correctness of the algorithm in the case when~$Q^\star$ contains no leaps.

Assume now that $Q^\star$ contains a parity-changing leap~$L^\star$. Let $a$ and $b$ denote the endpoints of~$L^\star$.
Let us consider the iteration of Algorithm~\ref{alg:SOP-tree} when the for-loop on line~\ref{algLine:for_loop} chooses $a$ and~$b$; 
we use the definitions on lines~\ref{algLine:leap}--\ref{algLine:def_path3}.
Note that $P_s$ and $P_t$ exist, since $Q^\star \setminus L$ decomposes into two vertex-disjoint paths from $\{s, t\}$ to $\{a, b\}$. 
Observe that $S_1 \triangle S_2=C$, and since~$L$ is parity-changing we have that $|C|$ is odd. Hence, exactly one of~$S_1$ and $S_2$ is an odd $(s,t)$-path.
We claim that $w(S_1) \leq w^\star$ and $w(S_2) \leq w^\star$.
Note that if our claim holds, then Algorithm~\ref{alg:SOP-tree} stores a path of weight at most $w^\star$ on line~\ref{algLine:def_path3}, and hence the correctness of the algorithm follows.
Therefore,  to show that Algorithm~\ref{alg:SOP-tree} is correct, it suffices to prove that
\begin{equation}
\label{eqn:claim-3cases}
\max\{w(S_1),w(S_2) \} \leq w^\star.
\end{equation}

Let $Q^\star_s$ and $Q^\star_t$ be the two vertex disjoint paths formed by $Q^\star \setminus L^\star$. Since $Q^\star_s$ and $Q^\star_t$ lead from $\{s,t\}$ to~$\{a,b\}$,
the definition of~$P_s$ and~$P_t$ implies $w(P_s)+w(P_t) \leq  w(Q^\star_s)+w(Q^\star_t)$.
Moreover, since $L$ is a minimum-weight parity-changing leap with endpoints~$a$ and~$b$, we also know 
\begin{equation}
\label{eqn:claim-twopaths}
w(P_s)+w(P_t)+w(L) \leq w(Q^\star_s)+w(Q^\star_t)+w(L^\star) = w(Q^\star)=w^\star.
\end{equation}   

We distinguish between three cases to prove inequality~(\ref{eqn:claim-3cases}).

{\bf Case A}: $x$ and $y$ are both on~$T^-[a,b]$. W.l.o.g., the vertices $a, x,y,b$ follow each other in this order on~$T^-[a,b]$, as otherwise we switch the names for~$a$ and~$b$. Note that $P_s \setminus P'_s$ and $P_t \setminus P'_t$ are two vertex-disjoint paths leading from $\{x,y\}$ to~$\{a,b\}$, so
by statement~(2) of Lemma~\ref{lem:T-min-pathlength}, we have 
\begin{equation}
\label{eqn:claim-caseA}
    w(P_s \setminus P'_s)+w(P_t \setminus P'_t) \geq w(T^-[x,a] \cup T^-[b,y]).
\end{equation}

Observe also that among the paths~$S_1$ and $S_2$, the one containing leap~$L$ has the larger weight since $w(T^-[x,y])<0$ and consequently $w(C \triangle T^-[x,y])>0$. Therefore, inequality~(\ref{eqn:claim-caseA}) together with  inequality~(\ref{eqn:claim-twopaths}) implies
\begin{align*}
    \max\{w(S_1),w(S_2)\}
    \, &= \, w(P'_s)+ w(T^-[x,a])+w(L)+w(T^-[b,y])+w(P'_t)  \\
    \, & \leq \, w(P'_s)+w(P_s \setminus P'_s) + w(L) +w(P_t \setminus P'_t) +w(P'_t) \\
    \, &= \, w(P_s)+w(L)+w(P_t) \\
    \, &\leq \, w^\star.
\end{align*}

{\bf Case B}: $x$ is on~$T^-[a,b]$, but $y$ is on~$L$. We may assume that $x \in V(P_s)$ and $y \in V(P_t)$, as otherwise we can switch the names of~$s$ and~$t$.
Next, we may also assume that $P_s$ ends at~$b$, and $P_t$ ends at~$a$, as otherwise we can switch the names of~$a$ and~$b$. 
Furthermore, we may even assume that $a \in V(S_1)$ 
and $b \in V(S_2)$, as otherwise we can switch the names of~$S_1$ and~$S_2$. 
This means that $P_s \setminus P'_s$ is an $(x,b)$-path, and $P_t \setminus P'_t$ is an $(a,y)$-path; recall that they are vertex-disjoint as well, by the definition of~$P_s$ and~$P_t$. Hence, $(P_s \setminus P'_s) \cup L[y,b] \cup (P_t \setminus P'_t)$ is an $(x,a)$-walk, and moreover, it does not contain any negative-weight edge more than once (recall that $L \cap T^-=\emptyset$, since~$L$ is a leap). By statement~(1) of Lemma~\ref{lem:T-min-pathlength}, this implies
\begin{equation*}
    w(P_s \setminus P'_s) + w(L[y,b]) + w(P_t \setminus P'_t) \geq w(T^-[x,a]).
\end{equation*}
From this, using also equation~(\ref{eqn:claim-twopaths}) we obtain that
\begin{align*}
w(S_1) \, &= \, w(P'_s)+w(T^-[x,a])+w(L[y,a])+w(P'_t) \\
\, &\leq \, w(P'_s)+w(P_s \setminus P'_s) + w(L[y,b]) + w(P_t \setminus P'_t)+w(L[y,a])+w(P'_t) \\
\, &= \, w(P_s)+w(L)+w(P_t) \\
\, &\leq \, w^\star. 
\end{align*}
To prove $w(S_2) \leq w^\star$ as well, we need to apply statement~(1) of Lemma~\ref{lem:T-min-pathlength} for the $(b,x)$-path $P_s \setminus P'_s$ to get 
\begin{equation}
\label{eqn:claim-caseB-s1}
   w(P_s \setminus P'_s) \geq w(T^-[b,x]).
\end{equation}
Furthermore, note that $(P_t \setminus P'_t) \cup L[y,a]$ is a closed walk that does not contain any edge of $T^-$ more than once, since $P_t \setminus P'_t$ is a path and $L \cap T^-=\emptyset$.
Thus, Lemma~\ref{lem:closed-walk} yields
\begin{equation}
\label{eqn:claim-caseB-s3}
   w(P_t \setminus P'_t) + w(L[y,a]) \geq 0.
\end{equation}
Putting inequalities (\ref{eqn:claim-twopaths}), (\ref{eqn:claim-caseB-s1}) and (\ref{eqn:claim-caseB-s3}) together we obtain
\begin{align*}
w(S_2) \, & = \, w(P'_s)+w(T^-[b,x])+w(L[y,b])+w(P'_t) \\
\, & \leq \, w(P'_s)+w(P_s \setminus P'_s) +w(L) - w(L[y,a])+ w(P'_t)\\
\, & \leq \, w(P_s)+w(L) +w(P_t \setminus P'_t)+w(P'_t)\\
\, & = \, w(P_s)+w(L)+w(P_t) \\
\, & \leq \, w^\star.
\end{align*}
This shows inequality~(\ref{eqn:claim-3cases}) for Case B.

{\bf Case C}: Both $x$ and $y$ are on~$L$. 
W.l.o.g., assume that the vertices $a,x,y,b$ follow each other in this order on the leap~$L$, as otherwise we can switch the names of~$x$ and~$y$. We may also assume that $S_1=P'_s \cup L[x,y] \cup P'_t$ while $S_2=S_1 \triangle C$, as otherwise we can switch the names of~$S_1$ and~$S_2$.
Note that $P_s \setminus P'_s$ and $P_t \setminus P'_t$ are two paths from $\{x,y\}$ to~$\{a,b\}$ and they are vertex-disjoint. Thus, their union together with $L[x,y]$ yields a walk from~$a$ to~$b$ that contains no negative-weight edge more than once. Statement~(1) of Lemma~\ref{lem:T-min-pathlength} then implies
\begin{equation*}
    w(P_s \setminus P'_s) + w(L[x,y]) + w(P_t \setminus P'_t) \geq w(T^-[a,b]).
\end{equation*}
Hence, using also inequality~(\ref{eqn:claim-twopaths}), we obtain
\begin{align*}
w(S_2) \,&=\, w(P'_s)+w(L[x,a)]+w(T^-[a,b]) + w(L[b,y])+w(P'_t) \\
\, & \leq \,  w(P'_s)+w(L[x,a)]+w(P_s \setminus P'_s) + w(L[x,y]) + w(P_t \setminus P'_t) + w(L[b,y])+w(P'_t) \\
\, &= \, w(P_s)+w(L)+w(P_t) \\
\, &\leq \, w^\star.
\end{align*}

To show $w(S_1) \leq w^\star$, observe that the paths $P_s \setminus P'_s$ and $P_t \setminus P'_t$ together with $L[a,x]$ and $L[b,y]$ form either one or two closed walks, depending on whether $P_s$ ends at~$a$ or at~$b$. Moreover, no edge of~$T^-$ appears more than once on these closed walks.  Therefore, Lemma~\ref{lem:closed-walk} yields
$$ w(P_s \setminus P'_s) + w(P_t \setminus P'_t)+w(L[a,x])+w(L[b,y]) \geq 0.$$
Thus, we obtain
\begin{align*}
w(S_1)\, &= \, w(P'_s)+w(L[x,y])+w(P'_t) \\
\, &\leq \, w(P'_s)+w(L[x,y])+w(P'_t) + 
w(P_s \setminus P'_s) + w(P_t \setminus P'_t)+w(L[a,x])+w(L[b,y]) \\
\, &= \, w(P_s)+w(L)+w(P_t) \\
\, &\leq  \, w^\star. \qedhere
\end{align*}
\end{proof}

\begin{thm}
\label{thm:T-SOP-inP}
\SOP{} can be solved in strongly polynomial time, 
if the weight function is conservative and the set of negative edges forms a tree.
\end{thm}
\begin{proof}
By Lemma~\ref{lem:algoTREE-correctness} we know that Algorithm~\ref{alg:SOP-tree} is correct, so it remains to bound its running time. Let $n=|V|$ and $m=|E|$ denote the number of vertices and edges in the input graph~$G$.

Clearly, the {\bf for}-loop 
on line~\ref{algLine:for_loop} is iterated $\binom{n}{2}$ times. Performing lines~\ref{algLine:F_even}--\ref{algLine:def_path2even} can be done in $\mathcal{O}(mn+n^2 \log n)$ by Theorem~\ref{thm:parity-constr-oddpath-inP}.
For line~\ref{algLine:leap}, by Remark~\ref{rem:min_parity-changing_leap} we can obtain a parity-changing leap of minimum weight in strongly polynomial time, in fact, in time $\mathcal{O}(mn+n^2 \log n)$ \cite{schrijver-book}.

To perform line~\ref{algLine:vertex_disj_paths}, we rely on Remark~\ref{rem:disjoint-paths}. Lines~\ref{algLine:def_Ps}--\ref{algLine:def_path3} can be performed in $\mathcal{O}(n+m)$ time.
Hence, if solving an instance of the \DISP{} problem on $G$ on line~\ref{algLine:vertex_disj_paths} takes $F(n,m)$ time, then we obtain an overall running time of $\mathcal{O}(n^2 (mn + n^2 \log n + F(n, m)))$. 
Note that currently the only known upper bound on $F(n,m)$ is $F(n,m)=\mathcal{O}(n^9)$~\cite{schlotter2023twopaths},  
but probably this can be improved. 
\end{proof}

We remark that there is a strong connection between the \SOP{} and the \DISP{} problems. On the one hand, line~\ref{algLine:vertex_disj_paths} of our algorithm for \SOP{} relies on finding two vertex-disjoint paths from two source terminals to two sink terminals.
The need to use such an algorithm as a subroutine seems inevitable: if the input graph admits a unique parity-changing leap~$L$, then in order to find a shortest odd path between~$s$ and~$t$, we may need to find two vertex-disjoint paths that lead from~$s$ and~$t$ to two vertices of the odd cycle induced by~$L$. 

On the other hand, there is a very simple reduction from \DISP{} to \SOP{}. To see this, consider the equivalent formulation of \DISP{} where instead of two source and two sink terminals in~$G$, we are only given two terminals, $s$ and~$t$, and the task is to find two openly disjoint $(s,t)$-paths in~$G$ while minimizing their total weight. 
To solve this problem using an algorithm for~\SOP{}, we first make a copy $s'$ of~$s$ and a copy~$t'$ of~$t$, and connect the vertices adjacent to $s$ to $s'$, and the vertices adjacent to $t$ to $t'$. Then, subdivide each edge of this graph 
(halving the weights as well), 
and add a new edge~$tt'$ of weight~$0$. Let us denote by $G'$ this resulting graph.
We obtain our two $(s,t)$-paths by computing an odd $(s,s')$-path of minimum weight in graph~$G'$. Indeed, any pair of openly disjoint $(s,t)$-paths in~$G$ corresponds to an odd $(s,s')$-path of the same weight in~$G$, and vice versa.
This reduction shows that the connection between these two problems is bidirectional.

\section{FPT algorithms}
\label{sec:FPT}
In this section, we present two FPT algorithms for solving \SOP{} when the weight function is conservative. 
In Section~\ref{sec:FPT-w-negative-edges} we present two algorithms related to $E^-$  the set of edges of negative weight: the first one is regarding the cardinality of $E^-$ and the second one uses the size of the maximum matching in $G[E^-]$.
Later, in Section~\ref{sec:FPT-treewidth} we propose an FPT algorithm where the parameter is the treewidth of the input graph.

\subsection{Parametrization by negative edges}
\label{sec:FPT-w-negative-edges}
This section is divided in two. First, we show a simple FPT algorithm when the parameter is the number of negative edges. Then, we proceed to a more interesting parameter, namely the size of a maximum matching in the graph spanned by all negative-weight edges.

The key idea of the algorithms is to guess for each negative edge~$e$ whether $e$ has an even or odd sequence number on an optimal path. With this guess, we run an instance of the \SPaPrOP{} problem and obtain our solution. 

\subsubsection{The number of negative edges as parameter}
\label{sec:FPT-number-negative-edges}
We show that we can solve \SOP{} by an FPT algorithm with $|E^-|$ as the parameter.

\begin{thm}
\label{thm:fpt_by_negative_edges}
There is an FPT algorithm for \SOP{} with conservative weights when parameterized by the number of negative edges $|E^-|$, running in $\mathcal{O}\left(2^{|E^-|}(mn +n^2 \log n)\right)$ time.
\end{thm}
\begin{proof}
Let $(G,w,s,t)$ be our input for \SOP{} with $G=(V,E)$, such that $E^- \subseteq E$ denotes the set of all edges with negative weight. 
 We guess a partition $(E^-_{\even},E^-_{\odd})$ of~$E^-$ 
 so that $E^-_{\even}$ and $E^-_{\odd}$ contains all negative edges with even and odd sequence number in an optimal solution for~$(G,w,s,t)$, respectively. 
 Then we solve the \SPaPrOP{} problem with instance $(G,w,s,t, E^-_{\even}, E^-_{\odd})$ using Theorem~\ref{thm:parity-constr-oddpath-inP}, and return the path $Q$ obtained. 

 To check the correctness of this algorithm, let $Q^\star$ be an optimal solution for \SOP{}. 
  As $Q$ is an odd $(s, t)$-path, $w(Q) \geq w(Q^\star)$. Conversely, if our guesses are correct, then $Q^\star$ is an $(E^-_{\even},E^-_{\odd})$-constrained odd $(s,t)$-path, and thus $w(Q) \leq w(Q^\star)$. 
  
  For the running time, it is clear that there are $2^{|E^-|}$ possible guesses, and getting a solution for the  \SPaPrOP{} problem can be done in $\mathcal{O}(mn+n^2 \log n)$ time by Theorem~\ref{thm:parity-constr-oddpath-inP}. Hence, the total running time is as stated.
\end{proof}

\subsubsection{The size of a maximum matching on the negative edges as parameter}
\label{sec:FPT-matching-negative-edges}

Let us now generalize the algorithm of Theorem~\ref{thm:fpt_by_negative_edges}. 
Let $\mu(E^-)$ be the size of a maximum matching in the graph $G[E^-]$. 
Note that $\mu(E^-) \leq |E^-|$ is trivial, and it is also clear that $|E^-|$ can be arbitrarily large for any fixed value of~$\mu(E^-)$; to see this, consider the case when $E^-$ is a collection of $\mu(E^-)$ stars of arbitrary sizes. Thus, the parameter~$\mu(E^-)$ is substantially weaker than the parameter~$|E^-|$.
In the following, $\mu(E^-)$ is simply denoted by $\mu$. 

\begin{thm}
\label{thm:randomized_by_negative_matching}
There is a randomized FPT algorithm for \SOP{} with conservative weights that returns an optimal solution to the problem with probability at least $1 - e^{-1} > 0.632$ in $\mathcal{O}\left(2^{2\mu}(mn + n^2 \log n)\right)$ time. 
\end{thm}
\begin{proof}
Let $(G,w,s,t)$ be an instance of the \SOP{} problem  with conservative weights. Let $E^- \subseteq E$ denote the set of all edges with negative weight. 
We would like to get an instance of the \SPaPrOP{} problem to solve our problem. Again, as in Theorem~\ref{thm:fpt_by_negative_edges}, we guess a random partition $(E^-_{\even},E^-_{\odd})$ of~$E^-$ by placing each edge~$e \in E^-$ into $E^-_{\even}$ with probability $\frac{1}{2}$, and setting $E^-_{\odd} = E^- \setminus E^-_{\even}$.
Then we obtain an optimal solution $Q$ for the \SPaPrOP{} problem with instance $(G,w,s,t, E^-_{\even}, E^-_{\odd})$ using Theorem~\ref{thm:parity-constr-oddpath-inP}. 
We repeat this procedure $T$ times and pick the path~$Q$ with the minimum weight, where the value of $T$ is chosen later. 

The key observation is that if $Q^\star$ is an optimal solution of the \SOP{} problem, then the set~$Q^\star_\even \subseteq E^-$ of edges with an even sequence number in $Q^\star$ form a matching in~$G$, and similarly, the set~$Q^\star_\odd  \subseteq E^-$ of edges with an odd sequence number in $Q^\star$ form a matching as well. This implies $|Q^\star_\even| \leq \mu$ and $|Q^\star_\odd| \leq \mu$, so $Q^\star$ contains at most $2\mu$ negative-weight edges.
Hence, we need to guess at most $2\mu$ edges of~$E^-$ correctly.

The probability of getting a guess where $Q^\star_\even \subseteq E^-_{\even}$ is $\left(\frac{1}{2}\right)^{\mu}$. 
Similarly, the probability of getting a guess where $Q^\star_\odd \subseteq E^-_{\odd}$ is $\left(\frac{1}{2}\right)^{\mu}$. 
Since there might be several optimal paths, the probability of getting a good set $E^-_{\even}$ is at least $\left(\frac{1}{2}\right)^{2\mu}$. 
Thus, we return a wrong path $Q$ with probability at most
\begin{align*}
    \left(1- \left(\frac{1}{2}\right)^{2\mu}  \right)^T 
    \leq  \left(  e^{- \left(\frac{1}{2}\right)^{2\mu}} \right)^T, 
\end{align*}
where we use the inequality $1 + x \leq e^x$. This implies that we succeed in returning a minimum-weight path with probability $1 - \left(  e^{- \left(\frac{1}{2}\right)^{2\mu} \cdot T} \right) $. Setting $T=2^{2\mu}$ as the number of repetitions, our algorithm succeeds with probability at least $1 - e^{-1} > 0.632$.

The running time depends on the time we need to compute each path $Q$. By Theorem~\ref{thm:parity-constr-oddpath-inP} and the value of $T$, we obtain  $\mathcal{O}\left(2^{2\mu}(mn + n^2 \log n)\right)$ as overall running time.
\end{proof}

It is straightforward to derandomize the algorithm using a universal set.

\begin{thm}
\label{thm:fpt_by_negative_matching}
There is a deterministic FPT algorithm for \SOP{} with conservative weights 
when parameterized by $\mu$, running in $\mathcal{O}\left(f(\mu) \cdot (mn \log n + 
n^2 (\log n)^2)\right)$ time, where $f(\mu) = 2^{2\mu} \cdot \mu^{\mathcal{O}(\log \mu)}$.
\end{thm}
\begin{proof}
    Let $(G,w,s,t)$ be our instance of \SOP{}; we use the notation in the proof of Theorem~\ref{thm:randomized_by_negative_matching}.
    By Theorem~\ref{thm:derandom_universal_set}, one can construct an $(n,2\mu)$-universal set family~$\mathcal{U}$ of size $|\mathcal{U}|=2^{2\mu} \cdot \mu^{\mathcal{O}(\log \mu)} \cdot \log n$ 
    in time $\mathcal{O}(n \cdot|\mathcal{U}|)$.
    Hence, instead of guessing $E^-_\even$ randomly as in Theorem~\ref{thm:randomized_by_negative_matching}, 
    we can try all sets~$U \in \mathcal{U}$, and define $E^-_\even=U \cap E^-$. By the definition of an $(n,2\mu)$-universal set family we obtain that $\mathcal{U}$ contains at least one set~$U^\star$ such that $Q^\star_\even \subseteq U^\star \cap E^-$ and $Q^\star_\odd \cap U^\star=\emptyset$, which suffices for our algorithm to produce a correct output. 
    The total running time is $\mathcal{O}\left(|\mathcal{U}|\cdot (mn +n^2 \log n)\right)$.
\end{proof}

\subsection{Parametrization by treewidth}
\label{sec:FPT-treewidth}
In this section we show that \SOP{} can be solved in linear time on any graph whose treewidth is bounded by a constant. The main idea is to show that this problem can be described by monadic second-order formulas.
\begin{thm}
\label{thm:bounded-tw}
For each fixed integer~$k$, the \SOP{} problem  
can be solved in linear time on graphs with treewidth at most~$k$.
\end{thm}

Theorem~\ref{thm:bounded-tw} can be proved by using Theorem~\ref{thm:monadic_treewidth}. 
For this, we need the following key lemma.

\begin{lem}
\label{lem:MSO-formula}
There exists a monadic second-order formula~\textup{$\texttt{OddPath}_{s,t}(P)$} defined on a graph~$G$ containing two fixed vertices~$s$ and~$t$ which holds for a set~$P$ of edges in~$G$ if and only if $P$ is an odd $(s,t)$-path in~$G$. 
\end{lem}
\begin{proof}
To define $\texttt{OddPath}_{s,t}(P)$, we first define a series of simple formulas of monadic second-order logic
on our graph~$G=(V,E)$. 
Recall that there is an atomic symbol~$\texttt{Inc}(e,v)$ which holds exactly if $e \in E$, $v \in V$, and $e$ is incident to $v$ in~$G$. 

We start with formulas $\texttt{Deg0}(F,v)$ and $\texttt{Deg1}(F,v)$ which hold for an edge set $F \subseteq E$ and a vertex $v \in V$ if and only if
$v$ has degree~0 and~1, respectively, in $G[F]$:
\begin{align*}
\texttt{Deg0}(F,v) & \equiv \nexists e: e \in F \wedge \texttt{Inc}(e,v), \\
\texttt{Deg1}(F,v) & \equiv \exists e: e \in F \wedge \texttt{Inc}(e,v) \wedge \left(\forall f:(f \in F \wedge \texttt{Inc}(f,v)) \Rightarrow (e=f))\right).
\end{align*}
Using these, we can easily construct 
a formula $\texttt{Matching}(M)$ which holds if and only if the edge set~$M$ is a matching in~$G$:
    \[ \texttt{Matching}(M) \equiv \forall v: \texttt{Deg0}(M,v) \vee \texttt{Deg1}(M,v).\]

We next define formula $\texttt{PartitionE}(F,F_1,F_2)$ which holds for edge sets~$F$, $F_1$, and $F_2$ if and only if $(F_1,F_2)$ is a partition of~$F$: 
\begin{align*}
    \texttt{PartitionE}(F,F_1,F_2) \equiv {} & \bigl(\nexists e \in E:(e \in F_1 \wedge e \in F_2)\bigr) \\ 
    & \wedge \bigl(\forall e \in E: (e \in F_1 \vee e \in F_2) \Leftrightarrow  (e \in F)\bigr). 
\end{align*}
We proceed by defining the formula $\texttt{ConnectedE}(S)$ which holds if and only if the subgraph~$G[F]$ is connected for a given set~$F$ of edges:
\begin{align*} 
    \texttt{ConnectedE}(F) \equiv \forall F_1,F_2 : {} &\bigl( \texttt{PartitionE}(F,F_1,F_2) \wedge (\exists e_1 \in F_1) \wedge (\exists e_2 \in F_2) \bigr) \\
    & \Rightarrow \bigl(\exists v,e_1,e_2 : e_1 \in F_1  \wedge e_2 \in F_2 \wedge \texttt{Inc}(e_1,v) \wedge \texttt{Inc}(e_2,v)\bigr).
\end{align*}
        
We are now ready to define~$\texttt{OddPath}_{s,t}(P)$ as follows: 
\begin{align*}
    \texttt{OddPath}_{s,t}(P) \equiv {} & \texttt{ConnectedE}(P)\\
    & {} \wedge \bigl(\exists P_1, P_2: \texttt{PartitionE}(P,P_1,P_2) \wedge 
    \texttt{Matching}(P_1) \wedge \texttt{Matching}(P_2)
    \\
    & \hspace{60pt} {} \wedge \texttt{Deg1}(P_1,s) \wedge \texttt{Deg0}(P_2,s) \wedge
    \texttt{Deg1}(P_1,t) \wedge 
    \texttt{Deg0}(P_2,t)\bigr).
\end{align*}

To see the correctness of the formula, first assume that $S \subseteq E$ is an odd $(s,t)$-path in~$G$. Let $S_1$ and $S_2$ be the set of edges in~$S$ with an odd and even sequence number, respectively. 
Then (i) $(S_1, S_2)$ is a partition of~$S$,  
(ii) $S_1$ and $S_2$ are both matchings, 
(iii) $s$ and $t$ both have degree~1 in~$S_1$ and degree~$0$ in~$S_2$, and (iv) $G[S]$ is connected. Hence, $G$ satisfies $\texttt{OddPath}_{s,t}(S)$.

Conversely, suppose that there exists an edge set~$P$ such that $G$ satisfies $\texttt{OddPath}_{s,t}(P)$.
By definition, this means that $P$ is connected and can be partitioned into two matchings, $P_1$ and $P_2$, satisfying also that $s$ and $t$ both have degree~1 in~$P_1$ but have degree~0 in~$P_2$. 
Observe that a connected graph can only be partitioned into two matchings if and only if it is a path or a cycle. The degree constraints on~$s$ and~$t$ then ensure that $P$ is a path with $s$ and~$t$ as its endpoints, and moreover, it is an odd path (since $P_1$ is incident to both~$s$ and~$t$).
This proves the correctness of $\texttt{OddPath}_{s,t}(P)$.
\end{proof}

\begin{proof}[Proof of Theorem~\ref{thm:bounded-tw}]
By Lemma~\ref{lem:MSO-formula}, the \SOP{} problem can be expressed as a \emph{linear EMS extremum problem}~\cite{ALS91}, using the monadic second-order formula~$\texttt{OddPath}_{s,t}(P)$ with 
the evaluation function being the edge weight function~$w$, so that the evaluation term for~$P$ becomes $w(P)$.
Therefore, Theorem~\ref{thm:monadic_treewidth} implies that  for any class $\mathcal{K}$ of graphs of bounded treewidth, the value 
\[\min \left\{w(P): P\subseteq E, \ \texttt{OddPath}_{s,t}(P)\right\}
\]
can be computed in linear time for any graph $G=(V,E) \in \mathcal{K}$, edge weight function~$w\colon E \rightarrow \mathbb{R}$, and vertices $s,t \in V$, assuming that a tree-decomposition of~$G$ is given as part of the input.
Since it is possible to compute a tree-decomposition of width at most $2k+1$ of any $n$-vertex graph with treewidth at most~$k$ in $2^{\mathcal{O}(k)}n$ time by Theorem~\ref{thm:construct-treewidth}, the theorem follows.
\end{proof}

We remark that it is also possible to give a direct FPT-algorithm for \SOP{} parameterized by the treewidth of the input graph using dynamic programming techniques on a tree-decomposition. Although the naive method for solving \SOP{} on a graph with $n$ vertices and a tree-decomposition of width~$k$ would yield an algorithm running in~$k^{\mathcal{O}(k)}n^{\mathcal{O}(1)}$ time, 
it is possible to obtain an $c^k k^{\mathcal{O}(1)} n$ algorithm for some constant $c$ 
using the rank-based approach of Bodlaender et al.~\cite{BCKN15} in almost the same way  
as they do for the \textsc{Traveling Salesman} 
problem; see Appendix~\ref{sec:appendix-treewidth} for the details.

\section{Conclusion}
We have identified islands of tractability for the computationally hard \SOP{} problem. We gave a polynomial-time algorithm for the case when the set of negative-weight edges forms a tree. 
We developed an FPT algorithm for the problem when parameterized by the  number of negative-weight edges, or more generally, by the size of a maximum matching in the graph spanned by all negative-weight edges. We also proved that the \SOP{} problem is linear-time solvable on graphs of bounded treewidth. 

A natural direction for further research is to identify additional cases when the problem becomes computationally tractable. 
A starting point could be to generalize our results to the case when the negative-weight edges form $c \geq 2$ trees. Is \SOP{} solvable in polynomial time when $c$ is a constant? If so, is it fixed-parameter tractable with~$c$ as the parameter?

\paragraph{Acknowledgement.}
We would like to thank Kristóf Bérczi who provided a technical review of the manuscript and was part of the organizers of the 12th Emléktábla Workshop, Hungary, July 2022, where the collaboration of the authors was hosted and supported. We also would like to thank Jesper Nederlof for personal communication about their results in~\cite{BCKN15}.

\paragraph{Funding.}
Csaba Király was supported by the J\'anos Bolyai Research Scholarship of the Hungarian Academy of Sciences, by the \'UNKP-21-5 New National Excellence Program of the Ministry for Innovation and Technology, and by  the Hungarian Scientific Research Fund (OTKA grants FK128673 and PD138102). 
Lydia Mirabel Mendoza-Cadena was supported by the Lend\"ulet Programme of the Hungarian Academy of Sciences -- grant number LP2021-1/2021 -- and by the Ministry of Innovation and Technology of Hungary from the National Research, Development and Innovation Fund, financed under the ELTE TKP 2021-NKTA-62 funding scheme.
Ildik\'o Schlotter is supported by the Hungarian Academy of Sciences under its Momentum Programme (LP2021-2) and its J\'anos Bolyai Research Scholarship, 
and by the Hungarian Scientific Research Fund (OTKA grants K128611 and K124171). 
Yutaro Yamaguchi was supported by JSPS KAKENHI Grant Number 20K19743 and by the Osaka University Research Abroad Program.

\bibliographystyle{plain} 
\bibliography{odd}

\begin{appendices}
\section{Dynamic programming for bounded treewidth}
\label{sec:appendix-treewidth}
We assume that the reader is familiar with the technique and terminology developed by Bodlaender et al.\ in their paper~\cite{BCKN15}. We show the necessary tools for solving the \SOP{} problem when the graph has treewidth $k$, in a similar fashion as the dynamic programming approach for the \textsc{Traveling Salesmen} problem in~\cite{BCKN15}. Using the results in~\cite{BCKN15}, one can show that the running time is $c^k k^{\mathcal{O}(1)} n$, for some constant $c$.

Let us consider a nice tree decomposition $(\mathbb{T} ,\mathcal B)$ as defined in~~\cite{BCKN15}. Let $r$ be the root vertex of $\mathbb{T}$; we may assume that the root bag contains only vertices $\{s,t \} $, and we denote it by $B_{r}$. Given a node $x$ of $\mathbb{T}$, we denote its bag by $B_x$, and by~$V_x$ the set of all vertices contained in the bag of some descendant of~$x$. The subgraph~$G_x$ is then defined as $(V_x,E_x)$ where $E_x$ is the set of edges introduced in some descendant of~$x$. 

Observe that given a solution to our instance~$(G,w,s,t)$, that is, an odd $(s,t)$-path~$P$, the edges of~$P \cap E_x$ span a set of vertex-disjoint paths in the subgraph $G_x$. 
For the dynamic programming algorithm, we use $\mathbf{d} \colon B_x \to \{0,1,2 \}$ to encode the degree of each vertex in~$B_x$ in such a partial solution~$P \cap E_x$. 
Note also that $P \cap E_x$ naturally defines a perfect matching on those vertices of~$B_x$ that have degree~$1$ in~$P \cap E_x$. 
We encode this information in a partition $M \in \Pi _2 \left(  \mathbf{d}^{-1}(1) \right)$ where $\Pi_2(S)$ denotes the set of all partitions over a set $S$ in which each block has size~$2$. 
Finally, we encode the parity of $|P \cap E_x|$ in an integer~$p$. 

For each node $x \in V(\mathbb{T})$, we can now define the set $\mathcal{E}_x(\cdot)$ of partial solutions we aim to capture
and the table $A_x(\cdot)$ that we are going to compute at~$x$. 
More precisely, given a function
$\mathbf{d} \in \{ 0,1,2\}^{B_x}$, a partition $M \in \Pi _2 \left(  \mathbf{d}^{-1}(1) \right)$ and integer $p \in \{0,1\}$
we define
\begin{align*}
   \mathcal{E}_x(M, \mathbf{d}, p ) =  \biggl\{   & F \subseteq E_x  :
              \forall v: v\in B_x \Rightarrow \operatorname{deg}_F(v) = \mathbf{d}(v)  \\[-2pt]
            &\qquad \qquad\wedge \,  \forall v:v \in V_x \setminus B_x \Rightarrow \operatorname{deg}_F(v) \in \{ 0,2\} \\
            &\qquad \qquad \wedge \,  \text{$G_x[F]$ contains no cycle} \\
           &\qquad \qquad \wedge \,  \forall u,v \in B_x: \{ u,v\} \in M   \Rightarrow \text{$u$ and $v$ are connected in $G_x$[F]} \\
             &\qquad \qquad \wedge \,  |F| \equiv p \mod 2 \, \biggl\}; \\
    A_x(\mathbf{d},p) =  \biggl\{ &  \left( M, \min_{F \in \mathcal{E}_x(M, \mathbf{d}, p )} w(F)  \right) :  
            M \in \Pi _2 \left(  \mathbf{d}^{-1}(1) \right) \wedge \mathcal{E}_x(M, \mathbf{d},p) \neq \emptyset \biggl\}.
\end{align*}

Observe that the weight of a minimum-weight odd $(s,t)$-path can be found by taking the weight stored in the table $A_{r}(\mathbf{d},1)$ for the partition $\{\{s,t\}\}$, where $\mathbf{d}(s) = \mathbf{d}(t) = 1$ and $\mathbf{d}(v) \in \{ 0,2 \}$ for each $v \in B_r \setminus \{ s,t\}$. We consider the following recurrence for $A_x(\mathbf{d},p)$, depending on the type of~$x$. Note that we give the recurrence formulas using the operators over weighted partitions defined in~\cite{BCKN15}. Although Bodlaender et al.~\cite{BCKN15} use non-negative integer weights, their machinery also works for integer weights.

\paragraph{Leaf bag $B_x$.} We have $B_x = \emptyset$. 
Then $A_x(\cdot)$ is empty for all  possible values, except for
\begin{equation*}
    A_x(\emptyset,0) = \{ (\emptyset, 0)\}.
\end{equation*}

\paragraph{Introduce vertex $v$ with bag $B_x$.} Let $y$ be the child of $x$, that is, $B_x = B_y \cup \{ v \}$. For all $\mathbf{d} \in \{0,1,2 \}^ {B_x}$ and $p \in  \{ 0,1 \}$ we set 
\begin{equation*}
    A_x(\mathbf{d},p) =\begin{cases}
        A_y(\mathbf{d}\vert_{B_y},p ) &\text{if } \mathbf{d}(v) = 0, \\
        \emptyset & \text{otherwise,}
    \end{cases}
\end{equation*}
as $v$ has just been introduced, and so it has no neighbours in $G_x$; therefore, there is no change in the connectivity or in the parity of partial solutions.

\paragraph{Forget vertex $v$ with bag $B_x$.} Let $y$ be the child of $x$, that is, $B_y = B_x \cup \{ v \}$. For all $\mathbf{d} \in \{0,1,2 \}^ {B_x}$ and $p \in  \{ 0,1 \}$ we set
\begin{equation*}
    A_x(\mathbf{d},p) = A_y(\mathbf{d}[v \to 0],p ) 
    \, \cuparrow \, 
    A_y(\mathbf{d}[v \to 2],p ). 
\end{equation*}
Either $v$ is contained in an optimal solution~$P$ or not. If $v$ is contained in~$P$, then since $v$ cannot have neighbours outside $V_x$, the two neighbours of~$v$ in~$P$ must be in $V_x$. Hence, the degree of~$v$ in a partial solution must be~$0$ or~$2$. Again, there is no change in parity or in connectivity; however, we need to discard dominated partitions.

\paragraph{Introduce edge $e=uv$ with bag $B_x$.} Let $y$ be the child of $x$. We have $B_x = B_y$. For all $\mathbf{d} \in \{0,1,2 \}^ {B_x}$ and $p \in  \{ 0,1 \}$ we set 
\begin{align*}
    A_x(\mathbf{d},p) =& A_y(\mathbf{d},p )
    \, \cuparrow \\ 
     & \quad  \bigcuparrow \begin{cases}
            \emptyset & \text{if }\mathbf{d}(u) = 0 \vee \mathbf{d}(v) = 0 , \\
            \mathtt{glue}_w \left( uv, A_y (\mathbf{d}[u,v\to 0],1-p) \right) & \text{if }\mathbf{d}(u) = 1 \wedge \mathbf{d}(v) = 1 , \\
            \mathtt{proj}\left(  \{v\} , \mathtt{glue}_w \left( uv, A_y (\mathbf{d}[u\to 0,v\to 1],1-p) \right)\right)  & \text{if }\mathbf{d}(u) = 1 \wedge \mathbf{d}(v) = 2 , \\
            \mathtt{proj}\left(  \{u\} , \mathtt{glue}_w \left( uv, A_y (\mathbf{d}[u\to 1,v\to 0],1-p) \right)\right)  & \text{if }\mathbf{d}(u) = 2 \wedge \mathbf{d}(v) = 1 , \\
            \mathtt{proj}\left(  \{u,v\} , \mathtt{glue}_w \left( uv, A_y (\mathbf{d}[u,v\to 1],1-p) \right)\right)  & \text{if }\mathbf{d}(u) = 2 \wedge \mathbf{d}(v) = 2. \\
        \end{cases}
\end{align*}
Clearly, a partial solution may not use the edge $e$, and we keep track of such partial solutions by keeping all undominated entries from the table $A_y(\mathbf{d},p )$ computed for the child $y$. If $e$ is contained in a partial solution~$F$ for~$x$ compatible with~$\mathbf{d}$ and~$p$, then the degree of~$u$ in $F \cap E_y=F \setminus \{e\}$ is smaller by~$1$ than its degrees in~$F$, and the same holds for~$v$. 
In particular, $\mathbf{d}(u) = 0$ or $\mathbf{d}(v) = 0$ is not possible. Furthermore, $F$ and $F \cap E_y$ have different parities. The simple case is when both $ \mathbf{d}(u) = 1$ and $\mathbf{d}(v) = 1$: then we can safely add the edge $e$ and use the parity that differs from~$p$; 
note that this means inserting the block $\{u,v\}$ into the partition, which we can achieve by using the $\mathtt{glue}_w$ operation.
If $\mathbf{d}(u) = 2$ or $\mathbf{d}(v) = 2$, then we first need to glue the partitions containing~$u$ and~$v$, and then we need to use the $\mathtt{proj}$ operation to remove those vertices from the partition that have degree~$2$ in~$F$ but degree~$1$ in~$F \cap E_y$ (such vertices can be~$u$, $v$, or both). Note that the $\mathtt{proj}(\{u,v\},\cdot)$ operation filters out those partitions where projecting out $u$ and $v$ decreases the number of blocks; hence, entries of $A_y(\mathbf{d},p )$ with partitions containing $\{u,v\}$ will not be used in the case $\mathbf{d}(u) = \mathbf{d}(v) = 2$. This way we avoid creating cycles.

\paragraph{Join with bag $B_x$.} Let $y$ and $z$ be the children of $x$, that is, $B_x = B_y = B_z$; 
let us now denote the ``left'' degree vector for child~$y$ by  $\mathbf{l}$ and the ``right'' degree vector for child~$z$ by~$\mathbf{r}$. For all $\mathbf{d} \in \{0,1,2 \}^ {B_x}$ and $p \in  \{ 0,1 \}$ we set
\begin{equation*}
    A_x(\mathbf{d},p) = \displaystyle \mathop{\bigcuparrow}_{\substack{\mathbf{l} + \mathbf{r} = \mathbf{d}, \\ p_1 + p_2 \,\equiv \, p \! \mod 2}}  
        \mathtt{proj}\left( \mathbf{d}^{-1}(2) \setminus (\mathbf{l}^{-1}(2) \cup \mathbf{r}^{-1}(2)), \mathtt{join}(A_y(\mathbf{l},p_1 ), A_z(\mathbf{r},p_2 )) \right).
\end{equation*}
Observe that $\mathbf{l} + \mathbf{r} = \mathbf{d}$ is a vector summation as $\mathbf{l},\mathbf{r},\mathbf{d} \in \{0,1,2 \}^ {B_x}$. Combining $(M_1,w_1) \in A_y(\mathbf{l},p_1)$ and $(M_2,w_2) \in A_z(\mathbf{r},p_2)$ to get an entry for $A_x(\mathbf{d},p)$ can be done if and only if $M_1 \cup M_2$ is acyclic and we have the correct parity. Maintaining acyclicity is equivalent to asking that the vertices in $\mathbf{d}^{-1}(2) \setminus (\mathbf{l}^{-1}(2) \cup \mathbf{r}^{-1}(2))$, which are the vertices that have degree 1 in the partial solutions for $y$ and for~$z$ but degree~2 in the partial solution for~$x$, are connected to vertices in~$\mathbf{d}^{-1}(1)$ in the resulting partition $M_1 \sqcup M_2$. We use the projection operation for ensuring this.

\bigskip
\noindent
Using the above formulas and applying the framework of~\cite{BCKN15}, it is straightforward to obtain an algorithm for \SOP{} that runs in $c^k k^{\mathcal{O}(1)} n$ time for some constant~$c$ on graphs with treewidth at most~$k$, where $n$ denotes the number of vertices.
\end{appendices}
\end{document}